\documentclass[a4paper,fleqn]{cas-sc}
\usepackage[english]{babel}

\usepackage{proof}
\usepackage{hyperref}
\usepackage{inference}
\usepackage{mathrsfs}

\usepackage{float}
\usepackage{enumitem}
\usepackage{amsmath}
\usepackage{textcomp}
\usepackage{amssymb}
\PassOptionsToPackage{usenames,dvipsnames}{xcolor}
\usepackage{xcolor}
\usepackage{stmaryrd}
\usepackage{graphicx}
\usepackage{subcaption}
\usepackage{xspace}
\usepackage{dsfont}
\usepackage{tikz}
\usetikzlibrary{decorations.pathreplacing,calc,math,angles,quotes} 
\usetikzlibrary{shapes.geometric}
\usepackage{pdftexcmds}
\usepackage{Latex/macros}
\usepackage{Latex/switch-macros}
\usepackage{Latex/tikz-macros}

\usepackage[most]{tcolorbox}

\usepackage[T1]{fontenc}
\usepackage[utf8]{inputenc}
\usepackage[numbers]{natbib}

\newcommand{\gray}[1]{\textcolor{gray}{#1}}

\newcommand\EQUAL{\gray{\eqcirc}}

\makeatletter
\newcommand{\sbullet}{%
  \hbox{\fontfamily{lmr}\fontsize{.55\dimexpr(\f@size pt)}{0}\selectfont\textbullet}}

\newcommand{\sring}{%
  \hbox{\fontfamily{lmr}\fontsize{.55\dimexpr(\f@size pt)}{0}\selectfont\ensuremath{\circ}}}
\makeatother

\newcommand{\exprdex}[1][x]{e_{#1}}
\newcommand{\fix}[1][S_1]{\rond{\mathsf{fix}}_{#1}}
\newcommand{\FIX}[2]{\sol(\fix[#1](#2))}
\newcommand{\mult}{\cdot}

\newcommand{\Act}{\mathsf{Act}}
\newcommand{\Reac}{\mathsf{Reac}}
\newcommand{\Pro}{\mathsf{Pro}}
\newcommand{\Inh}{\mathsf{Inh}}

\newcommand{\Continue}{\mathcal{C}\tau}
\newcommand{\slopeset}{\Sset_d}
\newcommand{\slope}{{\s}_d}
\newcommand{\Absdiff}[2]{\mathsf{diff}_{#1}(#2)}

\newcommand{\Absfun}{I_{\slopeset}}
\newcommand{\Absimage}[1]{\Absfun(#1)}
\newcommand{\Stcoef}[2]{\mathsf{stoich}_{#1}(#2)}
\newcommand{\ADE}[1]{\mathsf{ADE}(#1)}

\newcommand{\rondcnextstep}[1][\s]{\rond{\mathsf{cnext}}^{#1}}
\newcommand{\rondnextstep}[1][]{\rond{\mathsf{next}}^{#1}}
\newcommand{\nextstep}[1][]{\mathsf{next}^{#1}}
\newcommand{\causalsolution}{\mathsf{csol}}
\newcommand{\fexpr}[1][e]{\mathsf{fun}_{[#1]}}

\newcommand{\labels}{\mathcal{L}}

\providecommand{\R}{\mathbb{R}}
\newcommand{\Fun}{\mathbb{F}}
\newcommand{\Rarith}{\mathbb{R}_{\mathsf{arith}}}
\newcommand{\Romega}{\mathbb{R}_{\Omega}}
\newcommand{\Nat}{\mathbb{N}}

\renewcommand{\Succ}[1][\slope]{\mathsf{succ}^{#1}}
\newcommand{\SuccE}[1][\delta]{\Succ[#1,\R]}
\newcommand{\SuccS}{\Succ[\s]}

\newcommand{\idd}{ \mathsf{id} }
\newcommand{\homo}{h_{\Sset}}

\newcommand\dom[1]{\mathsf{dom}({#1})}
\newcommand{\DomF}[1]{\mathcal{D}_{#1}}
\newcommand{\constr}[1]{\rond{\Phi}_{#1}}
\newcommand\rond[1]{\mathring{#1}}

\newcommand{\s}{\mathbb{S}}

\newcommand{\FVN}{U}

\newcommand{\Symbs}{\Omega}
\newcommand{\cSeg}{\mathcal{O}_f}
\newcommand{\struct}{\mathcal{S}}

\newcommand\Vars{\mathcal{V}}

\newcommand{\formulas}[2][\dotsymbs] {\mathcal{F}_{#1}(#2)}
\newcommand{\ATrans}[2]{\mathcal{T}_{#1}(#2)}
\newcommand{\Sset}{\mathbf{S}}

\newcommand{\sol}[1][\Fun]{\mathsf{sol}^{#1}}
\newcommand{\solrond}[1][\Fun]{\rond{\mathsf{sol}}^{#1}}
\newcommand{\Expr}[1][\Sigma_{arith}]{\mathsf{Expr}_{#1}}
\newcommand\Rels{\Sigma}
\newcommand\Consts{\mathbf{C}}

\newcommand{\Dplus}{\mathrm{D}^+}
\newcommand\ignore[1]{}

\newcommand\set{\mathsf{set}}
\newcommand{\crochet}[2]{[\![ #1 ]\!]^{#2}}
\newcommand\fv{\mathsf{fv}}
\newcommand{\simpl}{\mathsf{simpl}}
\newcommand{\Nsimpl}{N_{\simpl}}
\newcommand{\NAct}{N_{\Act}}
\newcommand{\NInh}{N_{\Inh}}
\newcommand{\asimpl}{\alpha}
\newcommand{\aAct}{\beta}
\newcommand{\aInh}{\gamma}
\newcommand{\rAct}{\mathsf{rAct}}
\newcommand{\rInh}{\mathsf{rInh}}
\newcommand{\restr}[2]{{
  \left.\kern-\nulldelimiterspace 
  #1 
  \right|_{#2} 
  }}

\newcommand{\sem}[2]{\llbracket {{#1}} \rrbracket^{{#2}}}

\newtheorem{example}{Example}
\newtheorem{definition}{Definition}
\newtheorem{theorem}{Theorem}

\newtheorem{property}{Property}
\newtheorem{lemma}{Lemma}
\newproof{proofsketch}{Sketch of Proof}
\newproof{proof}{Proof}

\begin{document}
	\let\WriteBookmarks\relax
	\def\floatpagepagefraction{1}
	\def\textpagefraction{.001}
	
	\shorttitle{}    
	
	\shortauthors{}  
	
	\title [mode = title]{Abstract Simulation of Reaction Networks}

	
	%
	
	\author[1]{Marie-Eva Fabri}[orcid=0009-0004-1551-241X]
	
	\ead{marie-eva.fabri@univ-lille.fr}
	
	\author[2]{Joachim Niehren}
	
	\author[1]{Sara Riva}
	
	\author[1]{Cristian Versari}

	
	\affiliation[1]{organization={Univ. Lille, CNRS, Centrale Lille, UMR 9189 CRIStAL},
		postcode={F-59000 Lille}, 
		country={France}}
	
	\affiliation[2]{organization={Univ. Lille, CNRS, Inria, Centrale Lille, UMR 9189 CRIStAL},
		postcode={F-59000 Lille}, 
		country={France}}
	
	\cortext[1]{Corresponding author}
	
	\begin{abstract}
		Reaction networks model reactions between a finite set of species. 
		These networks can be associated with different semantics, depending on the type of analysis and the phenomena under study.
		The standard continuous semantics is given by a system of differential equations based on the kinetic expressions of the reactions.
		To simulate a network under this semantics, the full knowledge of the kinetic laws of each reaction and the initial concentrations of each species is necessary. 
		Since in empirical settings the quantitative information about the reactions can be partially or totally unknown, the challenge is to introduce new semantics that can still be applied. 
		In this direction, a recent approach in the state of the art concerning Reaction Networks proposes a qualitative abstraction that is too coarse to properly capture the time-course continuous behaviour.
		Starting from the ideas of this approach, in this paper we first introduce the \emph{causal continuous semantics} for Reaction Networks to capture their continuous-time dynamics, preserving the causality hidden inside each transition.
		Later, we introduce the \emph{differential sign semantics} to abstract in a qualitative way the behaviour of a system under the causal continuous semantics.
		We show that our new method, based on abstract interpretation, yields appropriate Boolean transition graphs that refine those provided by the previous approach.	
	\end{abstract}

	
	\maketitle
	
	\section{Introduction}

Reaction networks constitute a standard formalism for modeling biological systems. The mathematical theory of chemical reaction networks was introduced in the late 1970s, independently by Feinberg~\cite{feinberg_complex_1972} and by Érdi and Tóth~\cite{erdi_mathematical_1989}.
A reaction network is defined by a finite set of species together with a finite set of reactions. Each reaction consumes and produces multisets of species and is associated with a kinetic expression. 
Reaction networks can be given different semantics. The continuous dynamics of a reaction network is given by its deterministic semantics, i.e. a system of ordinary differential equations (ODEs) based on kinetic expressions of the reactions \cite{segel_modeling_1984}. The stochastic semantics interprets the model as a Markov chain where kinetic expressions provide the transition rates \cite{gillespie_exact_1977}. 
The continuous semantics can be obtained from the stochastic one (by mathematical limit operations where the number of molecules tends to infinity and the time steps tend to zero) under several assumptions. 
The discrete and the Boolean semantics, in contrast, are based on a qualitative interpretation of the model discretising time and species amounts \cite{gillespie_exact_1977}\cite{chaouiya_petri_2007}. All these semantics can be useful depending on the type of analysis and phenomena but, in empirical settings, where the quantitative information about the reactions of a system is partially or totally unknown,
only the discrete and Boolean semantics can be applied. The most extreme case of missing information is represented by partial reaction networks, where not only exact information about concentrations and kinetics parameters is not available, but also the kinetic expressions of the reactions are partially unknown. 

The stochastic, discrete and Boolean semantics are related through a well-defined abstraction hierarchy~\cite{fages_abstract_2008}. Their comparison with the continuous semantics is notoriously difficult.
It is indeed challenging to generate simulations approximating the continuous-time behaviour of the system and the synchronous characteristic of this semantics. 
Several approaches have been proposed to overcome this difficulty and to relate the continuous semantics to more qualitative interpretations of reaction networks.

Until now, the only known qualitative method that guarantees the inclusion of the continuous behaviour is the Most Permissive semantics \cite{pauleve_reconciling_2020}. 
However, this method is not specifically adapted to reaction networks under continuous semantics, making it overly general for this kind of systems. The qualitative approach proposed in \cite{niehren_abstract_2022} guarantees the inclusion of a variation of the continuous semantics that preserves causality but is too coarse to detect the difference between enzymatic activators and inhibitors, limiting its practical applicability. Nevertheless, this approach appears promising for the analysis of partial reaction networks.

In this paper, we tackle the challenge of defining a qualitative abstraction which captures the continuous dynamics by also taking causality into account, as well as by distinguishing between activations and inhibitions in enzymatic reactions. The key idea to achieve this goal is to distinguish the cause of a change in species concentrations from the change itself, the first being given by the derivative of the species in the corresponding ODE, the second being captured by the sign of its differential.
Thanks to this distinction, we introduce a qualitative interpretation of differential equations.

The paper is organised as follows. Section 2 introduces relational structures and reaction networks and presents two existing approaches to simulate reaction networks which are both based on {\em evolution constraints}, i.e. a set of rules approximating the solutions to the ODEs of the system.
Section 3 develops the abstraction of continuous functions.
In this part, we introduce {\em causal continuous predictions}, a modified set of solutions to the ODEs of the system consistent with the evolution constraints and representing a novel proposal of time-course modelling, which is both continuous and causality-aware.
In Section 4, we present the differential sign structure $\slope$, refining the sign structure presented in \cite{niehren_abstract_2022} by the addition of information about the differential of the variables. 
The abstraction of causal continuous predictions in this structure provides us with the first qualitative dynamics which is both aware of causality and formally derived from a continuous semantics.  
By applying abstract interpretation to evolution constraints, we obtain the {\em differential sign successor}, an abstract relation which we show to be readily computable by constraints solving and to provide a refined over-approximation of causal continuous predictions.
In the end, thanks to this approach, we obtain a computable abstraction of the dynamics that includes all causal continuous predictions of reaction network.
Throughout the paper, we will use the Lotka-Volterra predator–prey model as a running example to illustrate our constructions.

\begin{example}
	\label{ex:intro}
	The Lotka-Volterra predator–prey model \cite{bacaer_lotka_2011} consists of the 
	$$
	\begin{array}{l}
		\frac{dx}{dt} = \alpha x - \beta x y,\\
		\frac{dy}{dt} = \gamma x  y - \delta y.\\
	\end{array}
	$$
	It is used to describe the dynamics of a biological system in which two species interact (one as a predator and the other one as a prey). 
	In the equations above, $x$ is the population density of preys (for example rabbits) while $y$ is the population density of predators (for example foxes) and $\alpha,\beta,\gamma$ and $\delta$ are strictly positive parameters.
	
	Solutions of this kind of systems oscillate periodically except for two steady points: $(x=0,y=0)$ and $(x = \frac{\delta}{\gamma}, y = \frac{\alpha}{\beta})$.
	
\end{example}

	\section{Preliminaries}

\subsection{Relational structures}

We introduce the notion of relational structures defined over ranked signatures. This notion will serve as the semantic foundation for the interpretation of arithmetic expressions.

\smallskip 

A ranked 
signature is a set of symbols, each with an arity.
\begin{definition}[Ranked signature]\label{def: signature}
	A \emph{ranked signature} $\Sigma=((\Rels^{(n)})_{n>0},\Consts)$ consists
	of a set of symbols $\Consts$ (constants) and for each $n>0$ a set of symbols $\Rels^{(n)}$.
\end{definition}
The  symbols of $\Consts$ are interpreted as constants and those of $\Rels^{(n)}$ as operators of arity $n$.

We will only be interested in relational signatures $\Sigma$
such that the arity of the operators is at most $2$, i.e.
$\Sigma^{(n)}=\emptyset$ for all
$n> 2$.

\begin{definition}[Relational $\Sigma$-structure]\label{def: structure}
	A \emph{relational structure} $\struct$ over a signature $\Sigma$ consists of a domain together with an interpretation of each symbol in $\Sigma$ as either a relation or a constant.
	Formally, a \emph{$\Sigma$-structure} $\struct$ is a tuple
	$ ( \dom{S}, (\odot^\struct)_{\odot \in \Sigma^{(1)}\cup\Sigma^{(2)}} , (k^\struct)_{k \in \Consts })$, such that:
	\begin{itemize}
		\item the domain $\dom{\struct}$ of $\struct$ is a set;
		\item for $n \in \{1,2\}$ and any operator $\odot \in \Sigma^{(n)}$, the relation $\odot^\struct\subseteq \dom{\struct}^{n+1}$ has arity $n+1$;
		\item for any constant $k \in \Consts$, the value of the constant $k$ in $\struct$ is $k^{\struct} \in \dom{\struct}$.
	\end{itemize}
\end{definition}

Any relational interpretation of an operator $\odot \in \Sigma^{(n)}$ in the $\Sigma$-structure $\struct$ can be translated into a function $\odot^{\struct}_{\set}:\dom{\struct}^{n} \to \mathcal{P}(\dom{\struct})$ such that its image is a subset of the domain and, for any $(x_1,\dots,x_n)\in \dom{\struct}^{n}$,
$$
\odot^{\struct}_{\set}(x_1,\dots,x_{n}) = \{x\mid (x_1,\dots,x_n,x) \in
\odot^{\struct}\}.$$
A function can be seen as a particular kind of relation.
\begin{definition}[Partial and total functions]
	Let $\mathcal{R} \subseteq D^n$ be a $n$-ary relation over a set $D$, with $n \ge 2$. 
	If, for all $(x_1,\dots,x_{n-1}) \in D^{n-1}$, there exists at most one $y \in D$ such that $(x_1,\dots,x_{n-1},y) \in \mathcal{R}$,  
	then $\mathcal{R}$ is \emph{a partial function} from $D^{n-1}$ to $D$. We can write it as $\mathcal{R}: D^{n-1}\to D$.
	
	A partial function is said to be \emph{total} if each $(x_1,\dots,x_{n-1}) \in D^{n-1}$ is associated with exactly one $y \in D$.
\end{definition}
A classical arithmetic signature $\Sigma_{\mathrm{arith}}$ consists of the binary operators
$\Sigma_{\mathrm{arith}}^{(2)} = \{+, -, \times, /\}$,
the set of constants $\Consts^{\Sigma_{\mathrm{arith}}} = \mathbb{R}$,
and no unary operator, i.e. $\Sigma_{\mathrm{arith}}^{(1)} = \emptyset$.
Let us define the standard arithmetic structure over the real numbers.
\begin{definition}[Structure $\Rarith$]\label{def:Rarith}
	$\Rarith$ is the $\Sigma_{\mathrm{arith}}$-relational structure such that:
	\begin{itemize}
		\item $\dom{\Rarith} = \R$; 
		
		\item operators $\{+,-,\times,/\}$ are interpreted respectively as addition, subtraction, multiplication and division;
		\item each constant $k \in \R$ is associated with its value.
	\end{itemize} 
\end{definition}
\begin{example}\label{ex:operatorreal}
	According to Definition \ref{def:Rarith}, we have
	$3 +^{\Rarith}_{\set} 2 = \{5\}$
	and  $3 /^{\Rarith}_{\set} 0 = \emptyset$, which is equivalent to $$(3,2,x)\in +^{\Rarith} \Leftrightarrow x=5 ~~~ \text{and}~~~ \nexists y\in \R,~~ (3,0,y)\in /^{\Rarith}.$$
\end{example}

In the following, we will deal with solutions of differential equations in different kinds of structures extending $\Sigma_{\mathrm{arith}}$ with operators such as $\sqrt{(\cdot)}$, $exp(\cdot)$, $log(\cdot)$ that can be useful to describe the behaviour of the systems of interest. 
In general, any such operator can be included in the signature as long as its standard interpretation over real numbers is differentiable and the derivative is continuous on its whole domain of definition. 
For example, the standard interpretation of $log(\cdot)$ over $\R$ is the logarithmic function which is differentiable and has continuous derivative over its domain of definition $\R_+/\{0\}$. 
This property is captured by the following definition.
\begin{definition}[Standard signature $\Symbs$]\label{def: actualsignature}
	Let $\Symbs$ be a signature. We say that $\Symbs$ is a \emph{standard signature} if
	\begin{itemize}
		\item $\Symbs$ includes all operators and constants of $\Sigma_{\mathrm{arith}}$;
		\item for any operator $\odot \in \Symbs$, there is a standard interpretation over $\R$ that is differentiable and has continuous derivative over its domain of definition. 
	\end{itemize}
	
\end{definition} 
Let $\Symbs$ be a standard signature.
\begin{definition}[Structure $\Romega$]\label{def:Romega}
	We define $\Romega$ as the structure giving the standard interpretation of $\Symbs$ over $\R$. Specifically, for all arithmetical operators $\odot \in \Sigma_{arith}^{(2)}$, $\odot^{\Romega} = \odot^{\Rarith}$.
	
\end{definition}
Throughout this paper, the structure $\Romega$ will be the only one considered over real numbers. 
By abuse of notation, we will henceforth denote this structure by $\R$.

\smallskip

In order to rigorously define solutions of differential equations, we extend the signature $\Symbs$ to include a symbol representing the derivative.
We denote by $\dot{\Symbs}$ the signature given by $\dot{\Symbs} = \Symbs\cup\{\dot{}\}$, where $(\dot{~})$ is the symbol of arity $1$ chosen to represent the derivative. 
We are now able to define the  $\dot{\Symbs}$-structure of partial functions.
For each function $f$, we denote its domain of definition by $\DomF{f}$.

\label{page:fun}
\begin{definition}[The $\dot{\Symbs}$-structure  $\Fun$]\label{def:fun}
	We denote by  $\Fun$ the $\dot{\Symbs}$-structure such that:
	\begin{itemize}
		\item[(1)] $\dom{\Fun}= \{f : \R\to \R\}$ is the set of functions whose domain and codomain are subsets of $\R$;
		\item[(2)] for all constant $k\in \Consts$ and $x \in \R$, $ k^{\Fun}(x) = k^{\R}$;
		\item[(3)] for all $f,g \in \dom{\Fun}$, 
		$$(f,g)\in (\dot{~})^{\Fun}\Leftrightarrow 
		\left\{\begin{array}{l}
			\DomF{g} = \DomF{f},\\\forall x\in \DomF{g},~ g(x) = \frac{df}{dx}(x);
		\end{array}\right.$$
		
		\item[(4)] for all $\odot \in \Symbs^{(1)}$ and $f,g\in \dom{\Fun},$ 
		$$g \in \odot^{\Fun}_{\set}(f)\Leftrightarrow 
		\left\{\begin{array}{l}
			\DomF{g} = \DomF{f},\\\forall x\in \DomF{g},~ g(x) \in \odot^{\R}_{\set}f(x);
		\end{array}\right.$$	
		\item[(5)] for all $\odot \in \Symbs^{(2)}$ and $f,g,h \in \dom{\Fun},$ 
		$$h \in (f\odot^{\Fun}_{\set}g)\Leftrightarrow 
		\left\{\begin{array}{l}
			\DomF{h} =  \DomF{f}\cap\DomF{g},\\\forall x\in \DomF{h},~ h(x) \in f(x)\odot^{\R}_{\set}g(x).
		\end{array}\right. $$	
	\end{itemize} 
\end{definition}

Let $I \subseteq \R$ be a finite union of intervals and $\mathcal{C}^1(I)$ be the set of functions that are differentiable over $I$ with a continuous derivative.
For any $f\in \mathcal{C}^1(I)$, we say that $f$ is of class $\mathcal{C}^1$ over $I$. When the context is clear, by abuse of notation we sometimes 
write that $f$ is of class $\mathcal{C}^1$ by omitting $I$.
We denote as $\restr{f}{I} $ the function $f$ restricted to the domain $I \subseteq \DomF{f}$. 

\begin{example}\label{ex:fun}
	The function $5^{\Fun}$ satisfies $\forall x \in \R,\; 5^{\Fun}(x) = 5$.
	It is also the unique function $f$ such that $(3^{\Fun}, 2^{\Fun}, f) \in +^{\Fun}.$
	
	Let $\idd$ denote the identity function over $\R$ and let $f, g \in \dom{\Fun}$ be two functions such that
	$$
	(\idd, \idd, f) \in +^{\Fun}, \text{ and }
	(\idd, \idd, g) \in \mult^{\Fun}.
	$$
	The function $f$ is the total function that maps each input to twice its value, while $g$ is the total function that maps each input to its square. 
	
	Both solutions are unique because every operator in $\Fun$ is interpreted as a partial function.
	Notice that $ \idd, f, g \in \mathcal{C}^1(\R)$. 
	
	On the other hand, the image of $\idd /^{\Fun}_{\set} \idd$ is empty, since $0 \in \DomF{\idd}$ and $0/^{\R}_{\set} 0 = \emptyset$, which implies that there is no function $h \in \dom{\Fun}$ such that $(\idd,\idd,h)\in /^{\Fun}$. If we consider $I = \R \setminus \{0\}$ and $\idd^\star = \restr{\idd}{I}$, then we will have a solution $1^\star= \restr{1^{\Fun}}{I}$. In fact,
	$(\idd,\idd^\star,1^\star)\in /^{\Fun}.$
\end{example}

We now define the sign-based structure \cite{cousot_abstract_1977,sintzoff_calculating_nodate}, which was used to interpret qualitatively reaction networks consistently with their quantitative interpretation~\cite{niehren_abstract_2022}. 

\smallskip
\label{page:signset}
Let $\Sset$ be the sign set  $ \{-1,0,1\}$.
Let $\homo: \R \to \Sset$ denote the function that maps each real number to its sign. 
For all $x \in \R$,
$$\homo(x) = \left\{\begin{array}{lll}
	~~1 & \text{if} & x>0, \\
	~~0 & \text{if} & x=0, \\
	-1& \text{if} & x<0. \\
\end{array}\right.$$
\begin{definition}[Sign structure  \cite{cousot_abstract_1977,sintzoff_calculating_nodate}]\label{def:signstruct}
	We denote by $\s$ the $\Symbs$-structure such that: 
	\begin{itemize}
		\item $dom(\s) = \Sset$;
		\item for all constants $k\in \Consts$, $k^\s = \homo(k^\R)$;
		\item for all operators $\odot\in \Symbs^{(n)}$ of arity $n >0$, 
		$$\odot^\s = \{(\homo(x_0),\dots, \homo(x_n))\mid (x_0,\dots,x_n)\in \odot^\R\}.$$
	\end{itemize}
\end{definition}
\begin{example} 
	As shown in Example~\ref{ex:operatorreal}, we have $(3,2,5)\in +^{\R}$. As $\homo(5) = \homo(3) = \homo(2) = 1$, we deduce that $(1,1,1)\in +^{\s}$.
	More generally, as the addition of two positives numbers is a positive number, it holds that
	$1 +^{\s}_{\set} 1 = \{1\}$. Similarly, as the result of multiplying any number by zero is zero, we also have $1~\mult^{\s}_{\set}~0 =~\{0\}$. 
	However, the result of adding two numbers with opposite signs is not unique in $\s$. In fact, we have, for example:
	
	$$\begin{array}{l}
		(3,-5,-2)\in +^{\s},\\
		(3,-3,0)\in +^{\s}, \text{ and }\\
		(3,-2,1)\in +^{\s}.
	\end{array}$$ 
	Therefore, $(1,-1,-1)$, $(1,-1,0)$ and $(1,-1,1)$ are in $+^{\s}$, or equivalently, $1 +^{\s}_{\set} -1 = \{1,0,-1\}$.
\end{example}

$\s$ is constructed from $\R$ by means of the function $\homo$. This ensures that the interpretation of operators in $\s$ consistently corresponds to their interpretation in $\R$. Functions like $\homo$ are called homomorphisms.

\begin{definition}[$\Sigma$-homomorphism]
	We define a $\Sigma$-homomorphism between two $\Sigma$-structures $\struct$ and
	$\struct'$ as a function $h:\dom{\struct}\to \dom{\struct'}$ such that:
	\begin{itemize}
		\item $\forall k \in \Consts, $
		$ h(k^\struct) = k^{\struct'}; $ 
		%
		\item $\forall \odot \in \Rels^{(n)}$ of arity $n\in\Nat$,
		$$\{(h(x_1),\dots, h(x_{n+1}))\mid (x_1,\dots,x_{n+1} ) \in \odot^{\struct} \} \subseteq  \odot^{\struct'}.$$     
	\end{itemize}
\end{definition}

\begin{example}\label{ex:homo}
	$\homo$  is a $\Symbs$-homomorphism from $\R$ to $\s$. The function $f: \dom{\R} \to \dom{\Fun}$, such that $\forall x \in \R, f(x) = x^{\Fun}$, is a $\Symbs$-homomorphism from $ \R$ to $\Fun$ restricted to the signature $\Symbs$.
\end{example}
	\subsection{Expressions and formulas}

\subsubsection{Expressions}\label{para: expr}
Let $\Sigma$ be a signature, $\Vars$ be the set of variables, and $V\subseteq \Vars$.
The set of terms representing the $\Sigma$-expressions 
is given by the following syntax:
$$
e_1,\ldots, e_n \in \Expr[\Sigma](V) \ ::= \ k\quad\big\vert\quad x \quad\big\vert\quad
\odot(e_1,\dots,e_n),
$$
where $k \in \Consts$, $x \in V$ and $\odot$ is an operator in $\Sigma^{(n)}$.
An assignement is a map $\alpha:V\to \dom{\struct}$ (also denoted as an element of  $\dom{\struct}^V$). 
Given a $\Sigma$-structure $\struct$ and an assignement $\alpha$, 
For any term $e\in\Expr[\Sigma](V)$, we denote by $\crochet{e}{\alpha,\struct}$
its interpretation in $\struct$ according to  $\alpha$, defined inductively in the following way:
\begin{itemize}
	\item  for all constants $k \in \Consts$, $\crochet{k}{\alpha,\struct} = \{k^\struct\}$;
	\item for all variables $x\in V$, $\crochet{x}{\alpha,\struct} = \{\alpha(x)\}$;
	\item for all $\odot\in\Sigma^{(n)}$ (with $n>0$) and $e_1,\ldots e_n\in\Expr[\Sigma](V)$,
	$$\crochet{\odot(e_1,\dots, e_n)}{\alpha,\struct} = \bigcup_{(v_1,\dots v_n)\in \crochet{e_1}{\alpha,\struct}\times\dots\times\crochet{e_n}{\alpha,\struct}} \odot^S_{\set}(v_1,\dots,v_n).
	$$
\end{itemize}

\begin{example}
	We have $\crochet{x+y}{[x/3,y/3],\R} = \{6\}$ and
	$\crochet{x-y}{[x/3;y/3],\R} = \{0\}$. 
	Then, we can interpret a more complex expression like $$
	\begin{array}{ll}
		\crochet{\frac{x+y}{x-y}}{[x/3;y/3],\R}~	&=~~ \bigcup_{(v_1,v_2)\in\{5\}\times\{0\}}   ~~~v_1 /^\R_{\set} v_2 \\ 
		~ &=~~ 5/^\R_{\set}0\\
		~ &=~~ \emptyset.
	\end{array}$$ 
\end{example}

For any $\Symbs$-expression $e$ over a set of variables $V$,
let $\fexpr:~\R^V\to~\R$ be the function that maps an assignment  $\alpha:V\to \R$ to
the interpretation of $e$ in $\R$ with $\alpha$, i.e.\ 
$$\fexpr(\alpha) = y \quad\text{ when }\quad \{y\} =\crochet{e}{\alpha,\R }.$$

For example, it holds that $\fexpr[x\mult y]([x/3,y/2]) = 6$.

\subsubsection{Formulas}\label{para: formulas}
The set of $\Sigma$-formulas $\formulas[\Sigma]{V}$ is set
of terms defined by the grammar:

$$\phi, \phi' \in \formulas[\Sigma]{V} ::=  \quad\phi \wedge \phi' \quad\vert\quad e_1 \EQUAL e_2,$$
where $e_1$ and $e_2$ are elements of $\Expr[\Sigma](V)$.

The equality operator in set-valued expressions is different from the equality between sets, as it is interpreted 
non-deterministically, i.e. as a non-empty intersection. 

For all
$ e_1,~e_2~\in~\Expr[\Sigma](\Vars)$ and
$ \phi_1,\phi_2~\in~\formulas[\Sigma]{\Vars}$, we have:
$$
\begin{array}{l}
	\crochet{e_1 \EQUAL e_2}{\alpha,S} =
	\left\{
	\begin{array}{ll}
		\mathsf{False} ~~~&\textrm{if } \crochet{e_1}{\alpha,S}\cap \crochet{e_2}{\alpha,S} = \emptyset  \\
		\mathsf{True} & \textrm{otherwise};\\
	\end{array}
	\right.
	\\~\\
	\crochet{\phi_1 \wedge \phi_2}{\alpha,S} = \crochet{\phi_1 }{\alpha,S} \wedge
	\crochet{\phi_2}{\alpha,S}.
	\\
	
\end{array}
$$
We denote by $\fv(\phi)$ the set of free variables occurring in $\phi$.  
Given a structure~$\struct$, we denote $\sol[\struct](\phi)$ the set of assignments over $\fv(\phi)$ satisfying $\phi$, that is to say such that $\phi$ evaluates to $\mathsf{True}$ when interpreted in~$\struct$, i.e
$$
\sol[^\struct](\phi) = \   \{\alpha:\fv(\phi)\to \dom{\struct} \mid \sem{\phi}{\alpha,\struct} = \mathsf{True} \}.$$

\begin{example}
	Let us consider the $\Symbs$-formula
	$\phi := ~x+2 \EQUAL 1-2$. The set of solutions in $\R$ is $\sol[\R](\phi) = \{[x/-3]\}$.  In $\s$, we have:
	\begin{itemize}
		\item $\crochet{x+2}{[x/1],\s}\cap \crochet{1-2}{[x/1],\s} = \{1\}$ and consequently $[x/1] \in \sol[\s](\phi); $
		\item $\crochet{x+2}{[x/0],\s}\cap \crochet{1-2}{[x/0],\s} = \{1\}$ and consequently $[x/0] \in \sol[\s](\phi); $	
		\item $\crochet{x+2}{[x/-1],\s}\cap \crochet{1-2}{[x/-1],\s} = \{-1,0,1\}$ and consequently $[x/-1] \in \sol[\s](\phi).$	
	\end{itemize}
	Indeed, $\sol[\s](\phi) = \{[x/-1];[x/0];[x/1]\}$, even if the unique solution in $\R$ is negative.
	
\end{example}

	\subsection{Reaction networks}
\label{page:rn}
In this section, we formalise chemical reaction networks that have been introduced by Feinberg in \cite{feinberg_complex_1972}. 
A reaction network models the interactions between a set of chemical species by means of reactions, each one transforming a multiset of reactants into a multiset of products according to a kinetic expression that represents the rate of the reaction.

A multiset $St$ over a set of variables $S$ (that in our case represent species) is a function $S
\to \Nat$ that assigns a multiplicity to each species  $x \in S$. 
By abuse of notation, we can consider the domain of a multiset as the set of elements with non-null image: $\dom{St} = \{x\mid St(x)\neq 0\}$.
When $St(x) = 0$, we say that $x$ does not belong to $St$. 

\begin{definition}[Reactions]\label{def:reaction}
	A \emph{reaction} over a finite
	set of species $S$
	is a tuple $(r, \Reac, \Pro, e )$ where $r$ is the label of the reaction, $\Reac$ is the multiset of reactants, $\Pro$ the multiset of products
	and $e$, called kinetic expression, is an expression in $\Expr[\Symbs](S)$.
\end{definition}
We expect the functions induced by these kinetic expressions to satisfy a number of reasonable conditions. 
First, reactions cannot be active when one or more reactants are not available. 
Second, the functions should be positive.
\begin{definition}[Strict property \cite{fages_inferring_2015}]\label{def: strict}
	A reaction $(r, \Reac, \Pro, e)$ over a set of species $S$ is said to be \emph{strict} if and only if, for every assignment $\eta : S \to \R_+$, whenever there exists a species $x \in S$ such that $\eta(x) = 0$ and $\Reac(x) > 0$, we have $\fexpr(\eta) = 0$.
\end{definition}
\begin{definition} [Positivity property \cite{fages_inferring_2015}]\label{def: positivity}
	A reaction $(r, \Reac, \Pro, e)$ over a set of species $S$ is said to be \emph{positive} if and only if, $\fexpr$
	is defined on an open set including $\R_+^S$ and $\fexpr(\R_+^S)\subseteq \R_+^S$.
\end{definition}
According to these definitions, we can now formally define a reaction network. 
\begin{definition}[Reaction network]\label{def:rn}
	A \emph{reaction network} $N = (S,R)$ consists of a
	finite set of species $S$ and a finite set of strict and positive reactions $R$. 
\end{definition}
In the following, $\labels (R)~=~\{r|(r,\Reac,\Pro,e)\in R\}$ will denote the set of reaction labels.\label{page:rn suite}
Since each label is associated with a unique reaction,
we will use the notation $\Reac_r$ to denote the multiset of reactants in the reaction labelled $r$.
Similarly, we will use $\Pro_r$ and $e_r$, respectively, to denote the multiset of products and the expression of the reaction. 
\begin{example}\label{ex:rn}
	The Lotka-Volterra equations can be seen as a reaction network. The species of the system are $Ra$ (Rabbits) and $Fo$ (Foxes), which gives us a set of variables $S_{LV}=\{Ra,Fo\}$. For simplicity, we set all parameters of the reactions to $1$. Moreover, since the multiplicity of species is either 0 or 1, we will represent them by their domain only. The set of reactions $R_{LV}$ consists of the following reactions:
	\begin{itemize}
		\item the birth of rabbits $(b,\emptyset,\{Ra\},Ra)$ with kinetic expression $Ra$
		consumes nothing and produces rabbits;
		\item the hunt $(h,\{Ra\},\{Fo\},Ra\mult Fo)$  with kinetic expression $Ra\mult Fo$ consumes rabbits and produces foxes; 
		\item the death of foxes $(d,\{Fo\},\emptyset,Fo)$ with kinetic expression $Fo$ consumes foxes and produces nothing. 
	\end{itemize}
	The graphical representation of $LV$ is given in Figure \ref{fig:rn}.

\end{example}
\begin{figure}
	\centering
	\begin{tikzpicture}
		
		\node[draw,circle] (L)at(2.5,0) {Ra};
		\node[draw,circle] (R)at(7.5,0) {Fo};
		\node[draw,rectangle,label=above:
		$\mathit{Ra}$] (B)at(0,0) {b} ;
		\node[draw,rectangle,label=above: $\mathit{ Ra}\cdot \mathit{Fo}$] (H) at (5,0) {h};
		\node[draw,rectangle,label=above: $\mathit{Fo}$] (D)at(10,0) {d};
		\draw[->,>=latex] (B) -- (L);
		\draw[->,>=latex] (L) -- (H);
		\draw[->,>=latex] (H) -- (R);
		\draw[->,>=latex] (R) -- (D);
		
	\end{tikzpicture}
	\caption{\label{fig:rn} The reaction network $LV$, defined in Example~\ref{ex:rn}, represents the Lotka-Volterra model with all parameters set to $1$.  
		Species are represented by round-shaped nodes and reactions by square-shaped nodes.  
		Reactants are indicated by arrows oriented from species to reactions, while products by arrows from reactions to species.  
		If a multiplicity is greater than $1$, it is shown as a label above the corresponding arrow.  
		Labels above reactions denote their kinetic expressions. }
\end{figure}
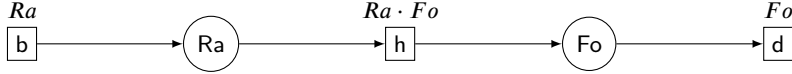
We denote the stoichiometry of the species $x$ in the reaction $r$ by $\Stcoef{r}{x}$.
The stoichiometry is defined as the difference between the multiplicity of $x$ as product and as a reactant, i.e. $$ \Stcoef{r}{x} =  \Pro_r(x) - \Reac_r (x). $$
\begin{example}
	
	In the table below, we provide the stoichiometry $\Stcoef{r}{x}$ for each species ($Ra$ and $Fo$) in each reaction ($b,h$ and $ d$) of the network $LV$. 	
	$$
	\begin{array}{|c|c|c|c|c|}
		\hline
		{x}\setminus{r}&~b~&~h~&~d~\\
		\hline
		
		Ra &  1 & -1 & 0\\
		\hline
		Fo& 0&1&-1\\
		\hline
	\end{array}	 	 
	$$
	We remark that, in the graphical representation of the network a missing arrow between a species $x$ and a reaction $r$ implies $\Pro_r(x) = \Reac_r (x) = 0$ and therefore $ \Stcoef{r}{x} = 0$. 
\end{example}
From a reaction network, we can derive a corresponding differential equation.

\begin{definition}[Systems of \textsc{ADEs} of reaction networks]
	\label{def:totalode}
	For any reaction network \( N = (S,R) \),
	we define the following system of algebraic differential equations   
	$$
	\ADE{N} = \bigwedge_{r \in \labels (R)} ( r \EQUAL e_r ) \bigwedge_{x \in S} (\dot{x} \EQUAL \exprdex)
	$$
	
	where $\exprdex = \sum_{r\in \labels (R)} \Stcoef{r}{x}\mult r$, while $\labels (R)$ is the set of reactions labels and $e_r$ the kinetic expression of the reaction labelled $r$. 
\end{definition}

The existence result for the above differential equation is provided in Appendix; see Proposition~\ref{prop: exists ADE} for a detailed statement and proof.

We observe that, even if the derivative of a species is defined as a sum over all reactions, its behaviour will only depend on the reactions where the stoichiometry of the species is not null. 

Let $\FVN=\fv (\ADE{N})$ denote the set of free variables of $\ADE{N}$.
In $\Fun$, the solutions of the differential equations defined above are of the form $f: \FVN\to (\R_+\to \R)$, i.e. functions mapping free variables to functions of $\dom{\Fun}$. 
It is however more convenient to view each such function $f$ as a function of time $g: \R_+\to (\FVN\to \R)$ mapping each time point $t \in\R_+$ to an assignment of variables $\FVN\to \R$, with no loss of information as $g(t)(x)=f(x)(t)$ for every $t$ and $x$. 

Therefore, with abuse of notation, when the second point of view is more convenient and the context is clear, we will interpret $f$ as an element of $\R_+\to (\FVN\to \R)$ by changing the order of its parameters and writing $f(t)(x)$ instead of $f(x)(t)$. 
We will also write $f_x(t)$ as a shortcut for $f(x)(t)$.

\begin{example}\label{ex:ODE-LV}
	We first express $\sum_{r\in \labels (R_{LV})} \Stcoef{r}{x}\mult r $ for each species $x$ in $S_{LV}$.
	Considering $Ra$, we have:
	\begin{equation}\label{eq: r1}
		1 \mult b + (-1)\mult h + 0 \mult d.
	\end{equation}
	Considering $Fo$, we have:
	\begin{equation}\label{eq: r2}
		0\mult b + 1\mult h + (-1)\mult d.
	\end{equation}
	For each structure considered before,  
	the expression (\ref{eq: r1}) is equivalent to $b-h$ and the expression (\ref{eq: r2}) is equivalent to  $h-d$.
	We now derive the ADE corresponding to $LV$, using these simplifications.
	\begin{equation}\label{eq:ADE-LV}
		\ADE{LV} =  (b \EQUAL Ra)\wedge (h \EQUAL Ra\mult Fo) \wedge (d \EQUAL Fo) \wedge (\dot{Ra} \EQUAL  b-h ) \wedge (\dot{Fo}\EQUAL h-d)
	\end{equation}
	
\end{example}

Since all $e_r$ are $\Symbs$-expressions, for any reaction network $N =(S,R)$, the system of equations $\ADE{N}$ is a $\dot{\Symbs}$-formula over the set of reactions and species, i.e. $\ADE{N}\in \formulas[\dot{\Symbs}]{\labels (R)\cup S}$, which can be interpreted in the structure $\Fun$.
The set of functions satisfying $\ADE{N}$ is $\sol(\ADE{N})$.The exact computation of solutions in  $\sol(\ADE{N})$, is not possible in general, but various numerical approximation methods have been proposed. \cite{euler_institutionum_1768,lapidus_numerical_1971}.
These methods are often based on discrete-time steps to approximate the evolution of species.
In such simulations, the derivative operator needs to be approximated to be able to compute the state of the system at the next time step.
Conceptually, this corresponds to replacing each derivative ($\dot{x}$) with a new free variable ($\rond{x}$) whose value depends on the current state of the system.
By applying this idea to the equations $\ADE{N}$, we obtain a system of constraints interpretable in the structure $\R$.

Formally, given a set of variables $V$, we denote by $\rond{V}$ a new set of variables in which,
for each $x \in V$, there is a corresponding variable $\rond{x}$ representing its derivative.
We now introduce the formula obtained by replacing derivatives with these variables as the evolution constraints.
\begin{definition}[Evolution constraints]\label{def: constraints}
	The set of evolution constraints, for a reaction network $N=(S,R)$, is 
	$$\constr{N} = \bigwedge_{r\in \labels (R)} (r\EQUAL e_r) \bigwedge_{x\in S} (\rond{x} \EQUAL \exprdex),
	$$
	where $\exprdex = \sum_{r\in \labels (R)} \Stcoef{r}{x}\mult r$, while $\labels (R)$ is the set of reactions labels and $e_r$ the kinetic expression of the reaction labelled $r$. 
\end{definition}

\begin{example}\label{ex:ODE-LV-bis}
	The evolution constraints of the LV model are:
	\begin{equation}\label{eq:constr-LV}
		\constr{LV}= (b \EQUAL Ra)\wedge (h \EQUAL Ra\mult Fo) \wedge (d \EQUAL Fo) \wedge (\rond{Ra} \EQUAL  b-h)\wedge(\rond{Fo}\EQUAL h-d).
	\end{equation}
	Let us consider an initial population of 10 rabbits and 2 foxes, i.e. $\alpha(Ra) = 10$ and $\alpha(Fo) = 2$. This gives $\alpha(\rond{Ra}) = -10$ and $\alpha(\rond{Fo}) = 18$ which implies a decrease in rabbits and an increase at foxes in the next time step.
\end{example}

	\subsection{Quantitative simulation}

A simple numerical simulation is Euler’s method \cite{euler_institutionum_1768,lapidus_numerical_1971} which discretises time and incrementally approximates the evolution of the variables of system based on their right-hand derivative, defined as the limit: $ \Dplus x(t) = \lim_{\delta\to0+} (x(t+\delta) - x(t))/\delta$. Such a limit can be approximated by choosing a sufficiently small  $\delta>0$, so that $ \Dplus x(t) \approx (x(t+\delta) - x(t))/{\delta}$. 
Thanks to this method, an approximated successor can be computed from any given state. 
We formalise this notion of successor for reaction networks.

\begin{definition}[Euler successor]\label{def: euler-succ}
	Let $N=(S,R)$ be a reaction network and $\delta>0$ be a time step. We define the successor $\SuccE(N)$ as the set of pairs $(\eta,\eta')\in \R^S\times\R^S$ such that:
	$$\exists \mu \in \sol[\R](\constr{N}), \forall x \in S, ~~\mu(x)=\eta(x) ~~\wedge~~ \eta'(x) \in \crochet{\rond{x}\mult \delta + x}{\mu,\R}.$$
\end{definition}

\smallskip Given a reaction network $N=(S,R)$ and a state $\eta \in \R^S$, there exists only one solution $\mu \in sol^{\R}(\constr{N})$  such that $ \forall x \in S$, $\mu(x)=\eta(x)$. 
In fact, since operators are interpreted as partial functions in $\R$, the evaluation of kinetic expressions at $\eta$ gives unique values.
Once these values are computed, they can be assigned to the variables representing reactions, which in turn allow the calculation of the values for the variables of the derivatives.
Therefore, even if the variables corresponding to the derivatives are free in $\constr{N}$, the successive state is uniquely defined.  
Thus, once a $\delta$ is chosen, it is possible to deterministically construct a simulation for any reaction network from any initial state.

\begin{example}
	We illustrate the computation of a successor pair $(\eta,\eta')$ for $LV$ with $\eta = [Ra/10;Fo/2]$ and a time step $\delta = 0.1$. 
	Let us recall the formula $\constr{LV}$ defined in Equation \eqref{eq:constr-LV}.
	$$\constr{LV} = (b \EQUAL Ra)\wedge (h \EQUAL Ra\mult Fo) \wedge (d \EQUAL Fo)\wedge (\rond{Ra} \EQUAL  b-h)\wedge(\rond{Fo}\EQUAL h-d)$$
	To compute an assignment $\mu$ over $\fv(\constr{LV})$ compatible with the initial state $\eta$ and the formula above, the following conditions must be met:
	\begin{itemize}
		\item since $\eta(Ra)=10$ and $\eta(Fo)=2$, we have $\mu(Ra)=10$ and $\mu(Fo)=2$;
		\item since  $\mu(Ra)=10$, $\mu(Fo)=2$, we have $\mu(b)=10$, $\mu(h)=20$ and $\mu(d)=2$;
		\item since $\mu(b)=10$, $\mu(h)=20$, $\mu(d)=2$, we have $\mu(\rond{Ra})=-10$ and $\mu(\rond{Fo})=18$.	
	\end{itemize}
	
	According to $\mu$, we can compute $\eta'$:
	$$	\begin{array} {l}
		\eta'(Ra) \in \crochet{\rond{Ra}\mult 0.1 + Ra}{\mu,\R} \Rightarrow \eta'(Ra) = 9\\
		\eta'(Fo) \in \crochet{\rond{Fo}\mult 0.1 + Fo}{\mu,\R} \Rightarrow \eta'(Fo) = 3.8.\\		
	\end{array}$$
	\label{ex:euler}
	The subsequent steps of this simulation are illustrated in Figure~\ref{fig:
		euler}, along with another simulation initiated from the  state $[Ra/1;Fo/1]$.
	
\end{example}
\begin{figure}
	\centering
	\includegraphics[width=5cm]{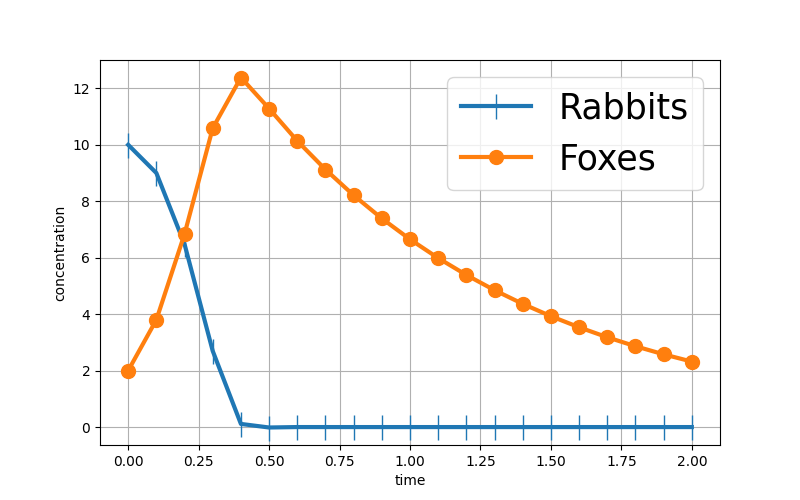}
	\includegraphics[width=5cm]{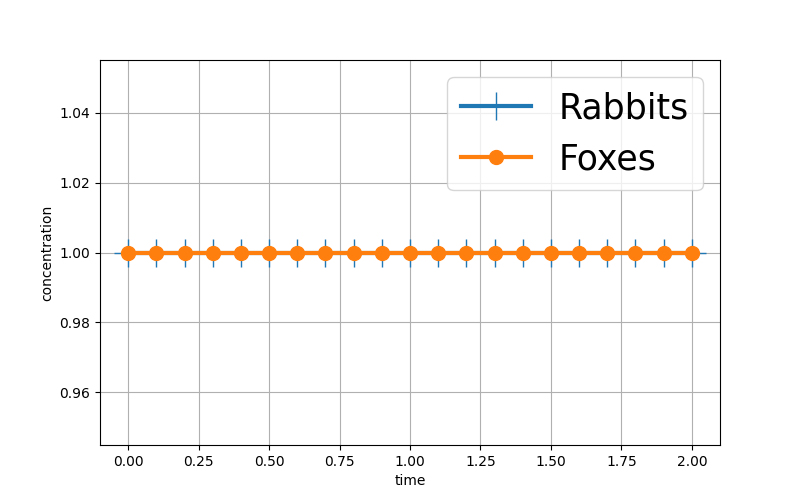}
	\caption{\label{fig: euler}Successive applications of the Euler successor on the reaction network $LV$ (Example~\ref{ex:rn}), with $\delta = 0.1$. On the left, the initial state is $[Ra/10; Fo/2]$. On the right, the initial state is $[Ra/1; Fo/1]$ and the two trajectories are superposed.} 
\end{figure}

In the scenario where all the quantitative information is available, numerical methods like Euler’s one are well suited for the analysis of reaction systems. 
However, in all the other scenarios, we need to exploit an abstraction approach to investigate the behaviour of the system.

	\subsection{Simulation graph over signs}
In \cite{niehren_abstract_2022}, the authors presented a sign-based abstract simulation for reaction networks applicable in the case of missing quantitative information about the initial state of the system. 
The simulation is based on the same idea of the Euler successor presented above, except for the fact that the constraints to model the evolution of the system are interpreted in the sign-based structure $\s$. 
Since the interpretation in $\s$ is non-deterministic, several states may be successors of a given state. 
Furthermore, we observe that the sign-abstraction for any $\delta > 0$ is always $1$. 
Thus, while there exist infinitely many Euler simulations depending on the choice of $\delta$, there is only one sign-based simulation of the Euler method. 

\smallskip

For a reaction network \(N = (S, R)\), it is possible to define a sign successor similar to Definition \ref{def: euler-succ}, as a relation between assignments over  \(S\). 

\begin{definition}[Sign successors]\label{def: sign-succ}
	For a reaction network $N = (S,R)$, we define the \emph{abstract successor relation} $\SuccS(N)$ as the set of pairs $(\alpha,\alpha') \in \Sset^\FVN\times \Sset^\FVN$, with $\FVN = S \cup\rond{S}\cup \labels (R)$, such that 
	$$ \alpha,\alpha' \in \sol[\s](\constr{N}) ~~\wedge~~ \forall x \in S,~~\alpha'(x)  \in \crochet{\rond{x} + x}{\alpha,\s} $$
	
	\noindent where $\rond{x}\in \rond{S}$ is the variable modelling the derivative of $x \in S$.
\end{definition}

Since the interpretation in the sign structure is non-deterministic, multiple successors may exist for any given state.
As a consequence, given an initial state, the behaviour of the system is not represented by a single trace but by a directed graph.

\ignore{However, with the sign structure, we will not always be able to deduce the sign of the derivative from the difference in species values between two consecutive steps. Indeed, if a species is assigned the value $1$ at two consecutive steps, we cannot determine the sign of its derivative. Therefore, we choose the sign successor to be a relation between solutions of \(\constr{N}\), which correspond to assignments over species variables, reaction variables, and derivatives of species variables.}

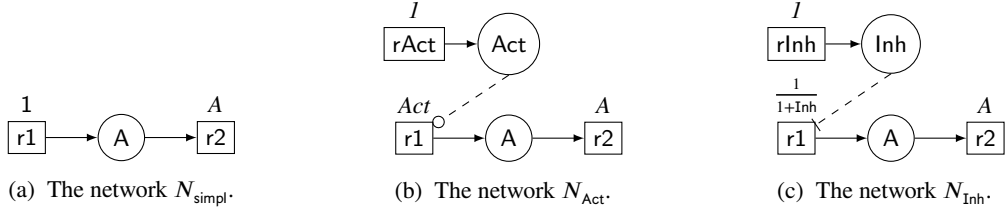
\begin{figure}
	\centering
	\begin{subfigure}{0.30\textwidth}
		\centering
		
		\begin{tikzpicture}[scale=2.5]
			\node[rectangle,draw,label=above:1] (r1Z) at (-3.5,0)	{r1};
			\node[circle, draw] (Z) at (-3,0) {A};
			\node[rectangle,draw,label=above:\it{A}] (r2Z) at (-2.5,0) {r2};
			\draw[->,>=latex] (r1Z) to (Z);		
			\draw[->,>=latex] (Z) to (r2Z);
		\end{tikzpicture}
		\caption{\label{fig:simpl} The network $\Nsimpl$. }
	\end{subfigure}
	\begin{subfigure}{0.30\textwidth}
		\centering
		\begin{tikzpicture}[scale=2.5]
			
			\node[rectangle,draw,label=above:\it{1}] (rAct) at (-1.5,0.5)	{rAct};
			\node[circle, draw] (Act) at (-1,0.5) {Act};
			\node[rectangle,draw,label=above:\it{Act}] (r1) at (-1.5,0)	{r1};
			\node[circle, draw] (A) at (-1,0) {A};
			\node[rectangle,draw,label=above:\it{A}] (r2) at (-0.5,0) {r2};
			\draw[->,>=latex] (r1) to (A);	
			\draw[->,>=latex] (rAct) to (Act);	
			\draw[->,>=latex] (A) to (r2);
			\draw[-o,>=latex, dashed] (Act.south) to (r1);
		\end{tikzpicture}
		\caption{\label{fig:Act} The network $\NAct$. }
	\end{subfigure}
	\begin{subfigure}{0.30\textwidth}
		\centering
		
		\begin{tikzpicture}[scale=2.5]
			\node[rectangle,draw,label=above:\it{1}] (rInh) at (0.5,0.5)	{rInh};
			\node[circle, draw] (Inh) at (1,0.5) {Inh};
			\node[rectangle,draw,label=above:$\mathit{\frac{1}{1+\Inh}}$] (r3) at (0.5,0)	{r1};
			\node[circle, draw] (B) at (1,0) {A};
			\node[rectangle,draw,label=above:\it{A}] (r4) at (1.5,0) {r2};
			\draw[->,>=latex] (r3) to (B);	
			\draw[->,>=latex] (rInh) to (Inh);	
			\draw[->,>=latex] (B) to (r4);
			\draw[-|,>=latex, dashed] (Inh.south) to (r3);		
		\end{tikzpicture}
		\caption{\label{fig:Inh} The network $\NInh$.}
	\end{subfigure}
	
	\caption{\label{fig:threeRN} The figure shows three different networks. On the left, $\Nsimpl$ consists only of the production and consumption of a species $A$. In the middle, $\NAct$ includes a regulatory interaction acting as an activator on the production of $A$. On the right, $\NInh$ includes a regulatory interaction acting as an inhibitor on the production of $A$.}
	
\end{figure}

\begin{example}
	By applying Definition \ref{def: sign-succ} to the network $\Nsimpl$ (Figure~\ref{fig:simpl}) starting from some initial state $\alpha \in \Sset^{\{A, \rond{A}, r1,r2\}}$, we can compute the set of sign successors for $\alpha$ (i.e, 
	$\{\alpha' \mid (\alpha,\alpha') \in \SuccS(\Nsimpl)\}$). We consider any initial state $\alpha$ where the species $A$ is present with a null derivative (i.e., $\alpha(A) = 1$ and $\alpha(\rond{A}) = 0$). 
	By Definition \ref{def: sign-succ}, we have that $\alpha,\alpha' \in \sol[\s](\constr{\Nsimpl})$ which implies: 
	$$
	\begin{array}{l}
		\alpha(r1) \in \crochet{1}{\alpha,\s} = \{1\}, \\
		\alpha(r2) \in \crochet{A}{\alpha,\s} = \{1\}. \\ 
	\end{array}$$
	Therefore, there exists only one such initial state $\alpha$:
	$$\alpha(x) = \left\{\begin{array}{ll}1 &\text{ if } x \in \{A,r1,r2\},\\
		0 &\text{ if } x = \rond{A}. \\
	\end{array}\right.$$
	All sign successors $\alpha'$ of $\alpha$ in $\Nsimpl$ should satisfy:
	\begin{enumerate}
		\item $\alpha'(A) \in \crochet{\rond{A}+A}{\alpha,\s}$, and \label{cond1}
		\item $\alpha' \in sol^{\s}(\constr{\Nsimpl}).$ \label{cond2}
	\end{enumerate}
	From condition \ref{cond1}, we have  $\alpha'(A) = 1$. 
	Moreover, if $\alpha' \in \sol[\s](\constr{\Nsimpl})$, then  $\alpha'(r1) = \alpha'(r2) = 1$. 
	Now that we have determined the values for the variables associated with the reactions, we can compute the possible values for the derivative. 
	According to the evolution constraints $\constr{\Nsimpl}$, we have
	$$\alpha'(\rond{A})\in \crochet{r1-r2}{\alpha',\s} = \{-1,0,1\}.$$
	Therefore, $\alpha$ has three different sign successors.
	In Figure ~\ref{fig:signtoocoarse}, we can see the simulation graph induced by the sign successor relation on $\Nsimpl$, projected on $A$ and $\rond{A}$ obtained from the initial state $\alpha$. 
	
	Applying Definition \ref{def: sign-succ} on networks $\NAct$ and $\NInh$ (respectively Figures~\ref{fig:Act} and \ref{fig:Inh}), given initial states $\beta \in \Sset^{\{A, \rond{A},\Act,\rond{\Act},\rAct, r1,r2\}}$ and $\gamma \in \Sset^{\{A, \rond{A},\Inh,\rond{\Inh},\rInh, r1,r2\}}$
	such that all species are present and $A$ is stable, we compute the set of their sign successors.
	There exists only one state $\beta$ of this form such that $\beta \in \sol[\s](\constr{\NAct})$ :
	$$\beta(x) = \left\{\begin{array}{ll}1 &\text{ if } x \in \{A,\Act,\rond{\Act},\rAct,r1,r2\},\\
		0 &\text{ if } x = \rond{A}. \\
	\end{array}\right.$$
	Similarly, there exists a unique state $\gamma$ of this form satisfying $\constr{\NInh}$. This state is such that $\gamma(\rond{A})=0$ and $\gamma(X)=1$ for all other variables $X$. Considering the fact that a sign successor of $\beta$ (resp. $\gamma$) must satisfy conditions \ref{cond1} and \ref{cond2} listed before, we obtain that $\beta$ (resp.  $\gamma$) has three sign successors, 
	i.e. states $\beta'$ (resp. $\gamma'$) such that, for all variables $X\neq\rond{A}$, it holds that $\beta'(X)=\beta(X)$ (resp. $\gamma'(X) = \gamma(X)$).
	Let us point out that, once projected on $A$ and $\rond{A}$, the graphs induced by the sign-successor relation for $\NAct$ and $\NInh$, from $\beta$ and $\gamma$ respectively, turn out to be equal to the graph obtained for $\Nsimpl$ and presented in Figure~\ref{fig:signtoocoarse}.

\end{example}

\begin{figure}
	\centering
	\begin{tikzpicture}[scale=0.7]
		\node[draw, ellipse] (B) at (-1,-1) {0,1};
		\node[ellipse, draw] (C) at (1,-1) {1,-1};
		\node[ellipse, draw] (D) at (1,1) {1,0};	
		\node[ellipse, draw] (E) at (-1,1) {1,1};		
		\draw[->,>=latex] (B) to (E);
		\draw[->,>=latex] (E) edge[loop above] ();
		\draw[->,>=latex] (D) edge[loop above] ();	
		\draw[<->,>=latex] (E) to (D);
		\draw[->,>=latex] (C) to (B);
		\draw[<->,>=latex] (C) to (D);
		\draw[<->,>=latex] (C) to (E);
		\path[->,>=latex] (C) edge[loop right] ();		
	\end{tikzpicture}
	\caption{\label{fig:signtoocoarse}The projection on $A$ and $\rond{A}$ of the graphs induced by the sign-successor relation on $N_{\simpl}$, $\NAct$, and $\NInh$ (Figure~\ref{fig:threeRN}), from any initial state that can be projected on the state $(1,0)$ in the graph. }
\end{figure}
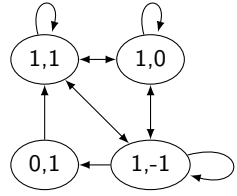
As illustrated by the previous example, the analysis provided by this method is too broad to reason about the dynamical behaviour of a reaction network.
In fact, whenever a species is both produced and consumed by active reactions, the method fails to capture any useful information about its dynamics.
Figure \ref{fig:signgraph} shows the outcome of the method on our running example, the $LV$ model.

\begin{figure}
	\centering
	\includegraphics[scale=0.45]{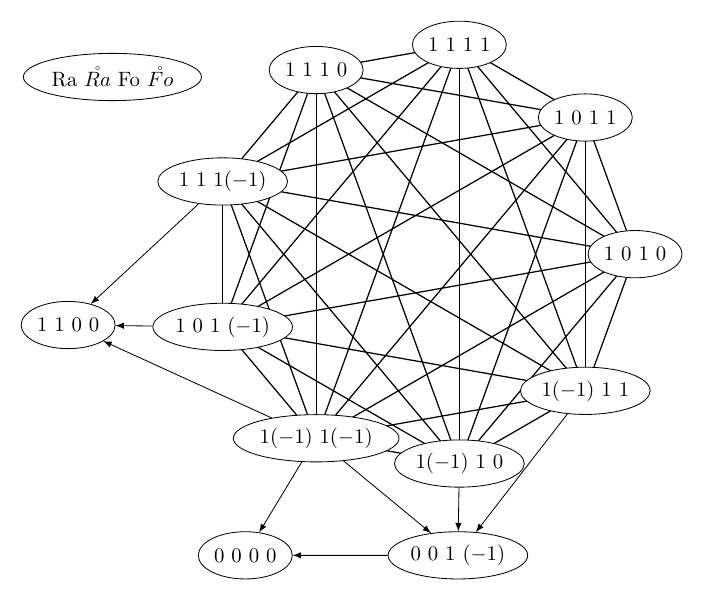}
	\caption{\label{fig:signgraph}
		The graph induced by the sign-successor relation on $LV$. The four values correspond respectively to $Ra$, $\rond{Ra}$, $Fo$, and $\rond{Fo}$ (as shown in the legend at the top left). Undirected edges represent symmetric relations between vertices. Since all states are successors of themselves, self-loops are omitted for clarity.}
\end{figure}

	\section{Abstraction of concrete simulations}

Since Euler-based simulations require complete knowledge of the system, and sign-based simulations are overly coarse, our goal is to significantly improve upon this latter approach. 
To do so, we exploit a more precise abstraction domain, while preserving causality and proving that the resulting model constitutes an abstraction of a continuous semantics. 
In \cite{niehren_abstract_2022}, the causality was managed through ad hoc functions which cannot be generalised to broader settings. 
In this section, we define a new abstraction function, introduce the refined domain, and propose a new causality-aware continuous semantics for the analysis of reaction networks.

\subsection{Abstraction of continuous semantics}
In this subsection, we define the abstraction of functions which are standard solutions of ODEs into transition relation.
Then, we apply it to the standard continuous semantics of reaction networks based on ODEs and show that it does not capture the underlying causality of the system. 

\smallskip

\label{page:abstractfunctions}
We start by defining the abstract transitions for functions over real numbers.

\begin{definition}[Abstract transitions]\label{def:seq}
	Let $V$ be a set of variables, $D$ a finite set, $f :\R_+ \to \R^V$ and $h: (\R_+\to\R_+^V)\to(\R_+\to D^V)$ two functions. 
	We define the \emph{$h$-abstract transition of $f$}, $\ATrans{h}{f}\subseteq (D^V)\times (D^V) $, as follows: 
	
	$$\ATrans{h}{f} = \left\{(\alpha,\beta)\mid 
	\begin{array}{l}
		\exists t_1,t_2 \in \R_+, t_1\leq t_2,\\ \alpha = h(f)(t_1) \wedge \beta = h(f)(t_2)\\ h(f) \text{ constant on } [t_1,t_2[ \text{ or } ]t_1,t_2]\\
		\text{ or } f \text{ non defined on } ]t_1,t_2[
	\end{array}\right\}.$$
\end{definition} 
When $(\alpha,\beta) \in \ATrans{h}{f}$, we say that $\beta$ is one next state of $\alpha$ for $f$ according to the $h$ abstraction.
\begin{example}
	The function $\fexpr[x^2]$ on $\R_+$ can be abstracted as a succession of states in $\Sset$.
	The first state is $0$ for $x=0$ and its only next state is $1$ which is then stable. 
	Indeed, $\homo \circ \fexpr[x^2](0)= 0$ and $\homo\circ \fexpr[x^2](t)= 1$, for any $t>0$.
	In Figure \ref{fig: seq}, we see the graphical representation of the function $\fexpr[x^2]$ 
	and  the graph induced by $\ATrans{h}{\fexpr[x^2]}$ where $h(f) = \homo\circ f$ 
	for all real functions $f$.	
\end{example}
\begin{figure}
	\centering
	\includegraphics[width=5cm]{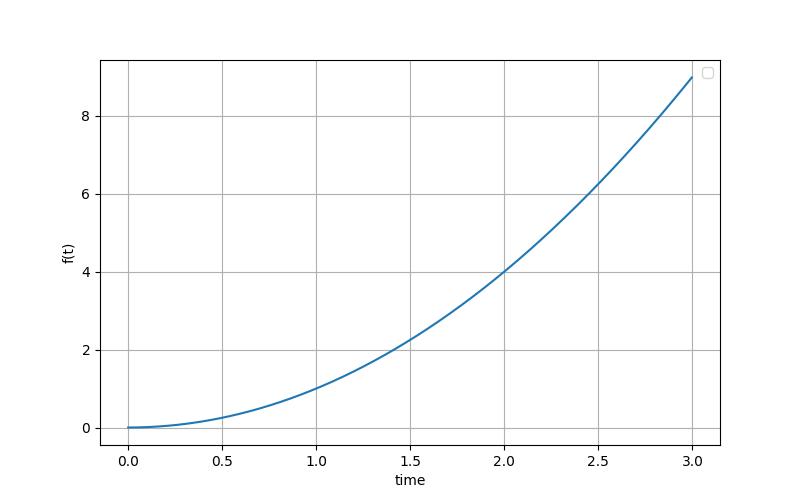}
	\includegraphics[width=5cm]{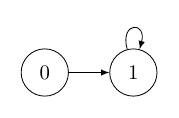}
	\caption{\label{fig: seq}  On the left, the function $x^2$ on $\R_+$, while on the right, the graph representing the transition relation $\ATrans{h}{\fexpr[x^2]}$ where $h(f) = \homo\circ f$ . }
\end{figure}
In order to obtain a more detailed abstraction of the behaviour of the system, we extend to differential equations the approach introduced in Definition~\ref{def: constraints} for evolution constraints, by introducing variables representing the derivatives.
For any equation formula $\phi$, let $S \subseteq \fv(\phi)$ be the set of variables $x$ such that the expression $\dot{x}$ appears in $\phi$. 
We define
$$
\solrond(\phi)=  \sol \left(\bigwedge_{x\in S} \rond{x}\EQUAL\dot{x}~\wedge  \phi\right).
$$
For a reaction network $N$, the subset $S$ defined above is precisely the set of variables representing the species of $N$ (see Definition~\ref{def:totalode}).

Once the derivatives are explicitly expressed, the transition graph abstracted from the real-valued solutions of a given differential equation can be defined as follows.

\begin{definition}[Temporal next relation]
	\label{def: next}
	Let $N=(S,R)$ be a reaction network, $D$ a finite set, and $h: (\R_+\to\R_+^V)\to(\R_+\to D^V)$ a function. 
	We define the \emph{temporal next relation of $N$ induced by $h$} as 
	
	$$
	\rondnextstep[h](N) = \bigcup_{f\in \solrond(\ADE{N})\cap \mathcal{C}^1(\R_+)}
	\ATrans{h}{\restr{f}{\R_+}}.
	$$
\end{definition}
The temporal next relation $\rondnextstep[h](N)$ is the union of the $h$ abstract transitions from all solutions of $\ADE{N}$. 
Throughout this paper, this concept is used exclusively with the abstraction $\homo$. 
For simplicity of notation, we therefore write $\rondnextstep[](N)$ instead of $\rondnextstep[\homo](N)$.

We now illustrate this definition with a simple example.

\begin{example}\label{ex: next}
	Let us consider the reaction network in Figure \ref{fig:exRN}. 
	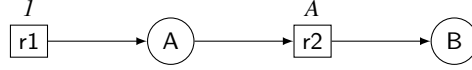
\begin{figure}
		\centering
		\begin{tikzpicture}[scale=2.5]
			\node[rectangle,draw,label=above:\it{1}] (r1) at (-1.75,0)	{r1};
			\node[circle, draw] (A) at (-1,0) {A};
			\node[circle,draw] (B) at (0.5,0) {B};
			\node[rectangle,draw,label=above:\it{A}] (r2) at (-0.25,0) {r2};
			\draw[->,>=latex] (r1) to (A);
			
			\draw[->,>=latex] (A) to (r2);
			\draw[->,>=latex] (r2) to (B);
		\end{tikzpicture}
		\caption{\label{fig:exRN}Example of a reaction network.}
	\end{figure}
	
	The network is composed of two species $A$, $B$ and two reactions: 
	$r_1$ producing $A$ with kinetic expression $1$, and $r_2$ consuming $A$ and producing $B$ with kinetic expression $A$. 
	This results in the following differential equation:
	\begin{equation}\label{eq:chaincausal}
		\ADE{RN}:= (r_1 \EQUAL 1) \wedge (r_2 \EQUAL A) \wedge (\dot{A} \EQUAL r_1 -r_2) \wedge (\dot{B} \EQUAL r_2).
	\end{equation}
	The unique solution $f$ of $\ADE{RN}$ on $\R_+$ satisfying $f_A(0)=f_B(0)=0$ is plotted on 
	the left-hand side of Figure~\ref{fig:chain}. 
	The sign abstraction of this solution $\ATrans{\homo}{f}$ consists of two transitions:
	$$\left([r_1/1;r_2/0;A/0;\rond{A}/1;B/0;\rond{B}/0], [r_1/1;r_2/1;A/1;\rond{A}/1;B/1;\rond{B}/1]\right) \text{ and}$$
	$$\left( [r_1/1;r_2/1;A/1;\rond{A}/1;B/1;\rond{B}/1],[r_1/1;r_2/1;A/1;\rond{A}/1;B/1;\rond{B}/1]\right).$$
	Without loss of information, we can project the states onto the variables $A,\rond{A},B,\rond{B}$ and represent any state with values ordered accordingly.
	Figure~\ref{fig:chain} shows on the right hand side the graph containing the abstraction of all solutions. 
	The abstraction of the solution $f$ identified before is the path starting at $0 1 0 0$.
	
	We can observe that $A$ and $B$ reach $1$ in the same state rather than sequentially. 
	Indeed for any time step $\epsilon >0$, the concentration of $B$ at time $\epsilon$ is strictly positive.
	Therefore, neither in the plot nor in the graph we observe that the presence of $B$ is caused by the presence of $A$.
	
\end{example}
\begin{figure}
	\centering
	\includegraphics[width=6cm]{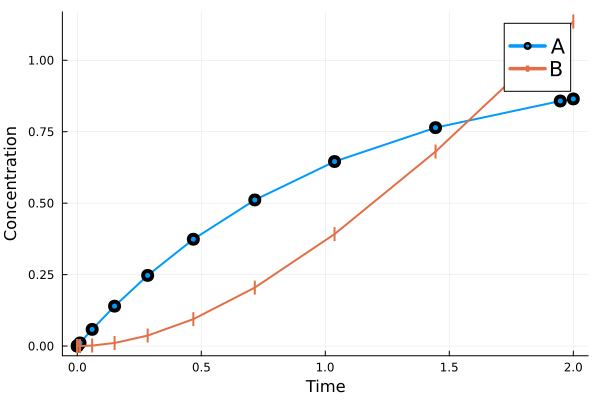}
	\includegraphics[width=4.5cm]{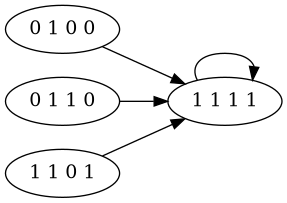}
	\caption{\label{fig:chain}The left-hand side shows the functions $f_A$ and $f_B$, where $f$ is the unique solution of $\ADE{RN}$ 
		(defined in Equation~\ref{eq:chaincausal}) satisfying $f_A(0)=0$ and $f_B(0)=0$. The right-hand side depicts the graph induced by $\rondnextstep(RN)$,
		projected onto $A, \rond{A}, B, \rond{B}$, with vertices ordered accordingly.}
\end{figure}
As illustrated in Example~\ref{ex: next}, abstractions of continuous simulations not always capture the causal structure producing the dynamics.
For this reason, in the next part, we will propose a new continuous semantics providing abstract transitions consistent with causality.

\subsection{Causal continuous semantics}

In this section, we define the criteria that allow us to identify continuous dynamics that are consistent with causality in reaction networks.  
To ensure this property over the whole set of solutions, we impose a notion of delay and we thereby introduce a new continuous semantics for reaction networks.

In order to capture the notion of causality implicit in species dynamics, 
when the variation of a species $A$ causes the variation of a species $B$, we must ensure that these variations appear one after the other in the trace.  
To achieve this result, we introduce the abstraction of the differential, which forecasts the direction of variation in species concentration.  

Let us recall that the differential $df$ is an infinitesimal change of $f(t)$ related to an infinitesimal change of time $dt$. 
\begin{definition}[Abstraction of the differential]\label{def: diff}
	Let $D$ be a finite set and $h: \R_+\to D$ a function. The abstraction of the differential by $h$ is the function $\Absdiff{h}{f}(x)$ such that:
	$$\Absdiff{h}{f}(t) = \lim_{\epsilon \to 0_+}  \Bigl( h \bigl(f(t+\epsilon)- f(t)\bigr)\Bigr). $$
\end{definition}

At this point, we apply Definition \ref{def: diff} using $\homo$ in order to obtain a new abstraction containing both the sign of the species concentration and the sign of the differential. 
\begin{definition}[Differential sign abstraction]\label{def:absimage}
	Let $f$ be a function with domain included in $\R$.
	The \emph{differential sign abstraction} of $f$, denoted $\Absimage{f}$, is defined as
	$$\Absimage{f}(t) = (\homo\circ f (t), \Absdiff{\homo}{f}(t))$$
	where $t\in\R$ and $\slopeset$ represents the set $\Sset\times\Sset$. 
\end{definition}

\begin{example}\label{ex:Absimage}
	To illustrate this definition, let us consider the application of $\Absdiff{\homo}{\cdot}$ to the square function $\fexpr[x^2]$. It holds that 
	$$  \Absdiff{\homo}{\fexpr[x^2]}(t) = \left\{\begin{array}{ccc}
		-1&\text{if}& t<0,\\
		~1&\text{if}& t=0,\\
		~1&\text{if}& t>0.\\
		
	\end{array}\right.$$
	It follows that
	$$  \Absimage{\fexpr[x^2]}(t) = \left\{\begin{array}{ccc}
		(1,-1)&\text{if}& t<0,\\
		(0,~1)&\text{if}& t=0,\\
		(1,~1)&\text{if}& t>0.\\
		
	\end{array}\right.$$

\end{example}

We can observe in Example \ref{ex:Absimage} that even though the sign of the derivative of $\fexpr[x^2]$ at $t=0$ is $0$,
its sign differential abstraction is $1$. This follows from the fact that, for any $\epsilon>0$, 
$\fexpr[x^2](\epsilon)>\fexpr[x^2](0)$. 
This shows how the sign of the derivative is not always predicting the future change in the behaviour of the function. However, the derivative is still relevant, because it highlights the causes (the influences of the other species) of such a change. 
Indeed, the derivative captures the cause of a change, while the differential captures the actual change of a species concentration. 
The inconsistency with causality can be related to the fact that the sign of the derivative does not necessarily coincide with the sign of the differential. 
We therefore introduce new tools to identify functions consistent with causality.

\ignore{With the $\ADE{\cdot}$ interpretation of reaction networks, as in Example \ref{ex: next}, interaction between species with each other determine their derivatives. However, the sign of this derivative does not always give the direction of variation. Indeed, in Example~\ref{ex: next}, since the first species $A$ is initially $0$, the derivative of $B$ is also $0$ at the initial time. However, the value of $B$ becomes positive immediately afterward. Therefore, even though its derivative is initially $0$, $B$ increases from $0$. 
	The difficulty is that the abstraction of the derivative given by the differential equation,
	does not necessarily coincide with the observed direction of evolution.  It is necessary to clearly define this latter concept, formally presented in Definition~\ref{def: diff}.

	The non-causal behaviour of $\rondnextstep$ comes from the mismatch between the dynamics inferred from contextual influence
	(i.e. those that determine the derivative) and the observed direction of evolution described by the differential abstraction. }
\begin{definition}[Causal match]\label{def:match}
	Let $V$ be a finite set of variables, $x \in V$ a variable, and $f:\R_+ \to \R^V$ a function.
	The function $f$ is a \emph{causal match for $x$ at $t \in \R_+$} if and only if the sign of differential matches the sign of the right-hand derivative, i.e.
	\begin{equation}\label{eq: causal}
		\Absdiff{h_\s}{f_x}(t) = \homo \circ \Dplus f(t).	
	\end{equation}
	Otherwise, we say that $f$ is a \emph{causal mismatch for $x$ at $t$}.
	For a subset $U \subseteq V$, if $f$ is a causal match for all $x \in U$ and for all $t \in \R_+$, 
	we say that $f$ is \emph{causal for $U$}.
	
\end{definition}
Causal mismatches may occur only when the derivative is null as stated by the following lemma. 
\begin{lemma}\label{lemma: mismatch}
	Let $V$ be a finite set of variables, $x \in V$ a variable, and $f:\R_+ \to \R^V$ a function.
	If the function $f$ is a causal mismatch for $x$ at $t$ then 
	$\Dplus f (t) = 0$.
\end{lemma}
\begin{proofsketch}
	Whenever the derivative is positive, the function increases locally,
	therefore it is a causal match at this time. 
	Similarly, a negative derivative implies a local decrease and thus a causal match. 
	However, when the derivative is zero, the function may increase or decrease immediately afterward.
	Hence, causal mismatch can arise only in the case where the derivative is zero.\qed
\end{proofsketch}
In order to obtain solutions which are causal for all species according to Definition \ref{def:match}, we will slightly relax the constraints imposed by $\ADE{N}$ 
by allowing species to temporarily delay their increase or decrease when their derivative is zero.
To achieve this, we consider differential equations that hold a subset of species fixed 
and allow the others to evolve according to the dynamics dictated by the reaction network. 
\begin{definition}[$\mathbf{ADE(\cdot)}$ with Fixed Variables]\label{def:fix}
	Let $N=(S,R)$ be a reaction network such that
	$$
	\ADE{N} = \bigwedge_{r \in R} r \EQUAL e_r \bigwedge_{x \in S} \dot{x} \EQUAL e_x.
	$$ 
	For any subset $S_1 \subseteq S$, we define a new differential equation: 
	$$
	\fix(N) =  \bigwedge_{r \in \labels (R)} r \EQUAL e_r \bigwedge_{x \in S/S_1} (\dot{x} \EQUAL e_x \wedge \rond{x}\EQUAL\dot{x})
	\bigwedge_{x \in S_1} (\dot{x} \EQUAL 0 \wedge \rond{x} \EQUAL e_x).\
	$$ 		
\end{definition}
The existence result for the above differential equation is provided in Appendix; see Proposition~\ref{prop: exists FIX} for a detailed statement and proof.

In order to ensure that the sign of the differential is zero and thereby prevent any mismatch, we formalise the notion of delay.
Intuitively, a delay corresponds to the entire time interval during which a species in $S_1$ is held constant using the differential equation introduced in Definition \ref{def:fix}, immediately following a potential causal mismatch. 
The equation $\fix(N)$ will be the one defining the behaviour of the system within the delay, and $\rond{x} \EQUAL e_x$ will be necessary to recover the original behaviour once the delay is over.
\label{page:delay}
\begin{definition}[Delay]\label{def: delays}
	Let $N=(S,R)$ be a reaction network, and $f: \R_+\to \R^{\labels (R)\cup S\cup \rond{S}}$ a function.
	A \emph{delay of $f$} (for some non-empty set $S_1\subseteq S$) is a maximal open connected set $\mathcal{O}\subset \R_+$ such that,
	for all  $\mathcal{O}_l$ open non-empty subset of $\mathcal{O}$, we have:
	\begin{itemize}
		\item $ \restr{f}{\mathcal{O}_l}\in \FIX{S_1}{N};$
		\item $ \restr{f}{\mathcal{O}_l}\notin \solrond(\ADE{N}).$
	\end{itemize}
\end{definition}

The delay $\mathcal{O}$ is maximal by inclusion, i.e. there is no open connected set strictly containing $\mathcal{O}$ satisfying the above constraints.
We denote by $\mathcal{O}_{f}$ the union of all delays of $f$.

\smallskip

By exploiting the notions of delay and mismatch, we can introduce the set of \emph{causal continuous predictions}, which represents the new causality-aware continuous semantics for the analysis of reaction networks that we anticipated before. 
Unlike solutions of $\ADE{N}$, these functions are causal for species with respect to $\Absimage{\cdot}$ thanks to the introduction of delays, which will be short enough to ensure that the functions remain unchanged according to $\Absimage{\cdot}$.

\begin{definition}[Causal continuous predictions]\label{def: ccpred}
	Given a reaction network $N=(S,R)$, the set of \emph{causal continuous predictions} for $N$, denoted by $\causalsolution(N)$,
	is the set of functions
	$f:\R_+\to \R^{S\cup\rond{S}\cup \labels (R)}$ such that:
	\begin{enumerate}
		\item\label{constr: continuity} \emph{[continuity]} 
		$f$ is continuous on a connected domain;
		
		\item\label{constr: small enough} \emph{[admissible delays]} for all delays $]t_1,t_2[$ of $f$ for $S_1\subset S$, we have:
		\begin{itemize}
			\item $\forall x\in S\cup\rond{S}\cup \labels (R)$, the abstract function 
			$ \Absimage{f_x}\text{ is constant on }[t_1,t_2[\text{ or }]t_1,t_2],$
			\item $\forall x \in S_1$, $f_{\rond{x}} (t_1) = 0$;
		\end{itemize}

		\item\label{constr: real ADE} \emph{[consistency]} for all $t \in \DomF{f}\setminus\mathcal{O}_{f}$,
		there exists an open set $\mathcal{O} \subset \DomF{f}$ containing $t$ such that  
		$$ \exists f' \in \solrond(\ADE{N}), \restr{f'}{\mathcal{O} \setminus \mathcal{O}_{f}} = \restr{f}{\mathcal{O} \setminus \mathcal{O}_{f}}; $$
		\item\label{constr: causal} \emph{[causality]} $f$ is causal (with respect to $\Absimage{\cdot}$) for the variables of the species.
	\end{enumerate} 
\end{definition}
The existence of causal continuous prediction is demonstrated in the Appendix, see Proposition~\ref{prop: exists CSOL}.
A causal continuous function evolves according to the dynamics induced by the reaction network, 
except for short delay intervals of $ \mathcal{O}_{f}$ introduced to ensure causality.
Outside these delays, the evolution is dictated by Constraint~\ref{constr: real ADE}.
Since $\R_+\setminus \mathcal{O}_{f}$ is not an open set, 
Constraint~\ref{constr: real ADE} cannot be expressed simply as the requirement that $f$ solves the differential equation on $\R_+\setminus \mathcal{O}_{f}$. 

\smallskip

We illustrate Definition \ref{def: ccpred} on a simple reaction network. 
\begin{figure}
	\centering
	\includegraphics[width=5cm]{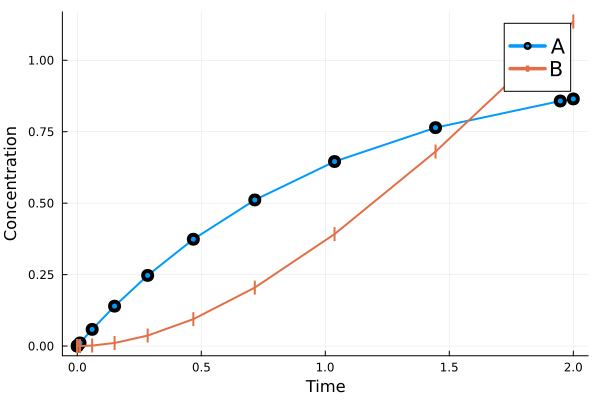}
	\includegraphics[width=5cm]{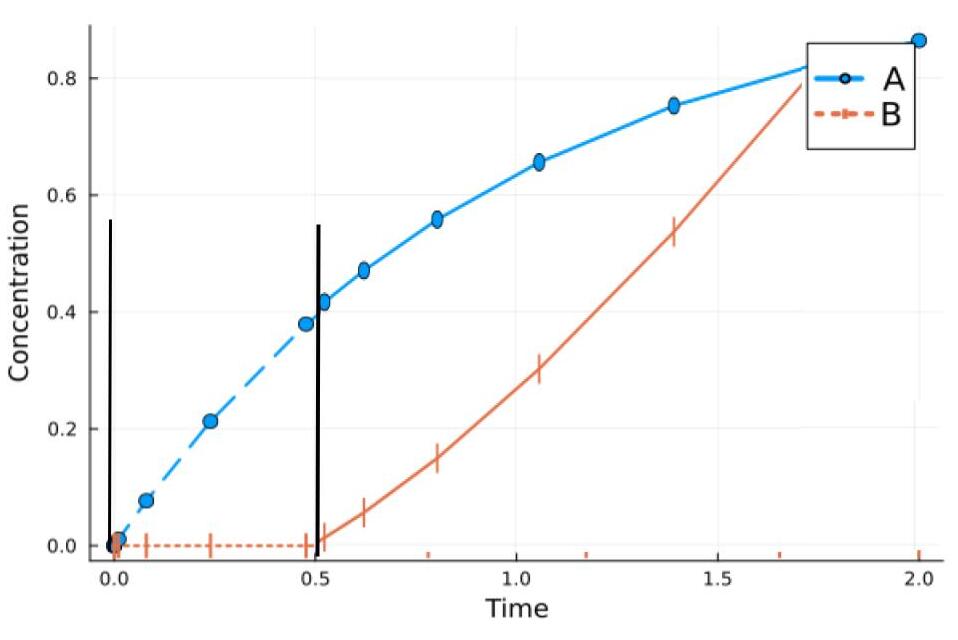}
	\caption{\label{fig:causal} On the left, the graphical representation of the unique solution of $\ADE{RN}$ defined in Equation~\ref{eq:chaincausal} of Example~\ref{ex: next},
		with initial conditions $A(0)=0$ and $B(0)=0$. On the right, the causal continuous prediction of Example~\ref{ex: causal} for $RN$. The dashed part represents the delay for $B$. }
	
\end{figure}

\begin{example}
	\label{ex: causal}
	
	We compute a causal continuous prediction for the reaction network of Example~\ref{ex: next} and the initial state such that $A(0)=0$ and $B(0)=0$. 
	Initially, the derivative of $B$ is equal to $0$. 
	However, as shown on the left of Figure~\ref{fig:causal}, 
	the solution of $\ADE{RN}$ with these initial conditions 
	satisfies $f_B(t)>0$ for all sufficiently small $t>0$. 
	Hence, the unique solution $f^0\in \sol(\ADE{RN})$ such that $f^0_A(0)=f^0_B(0)=0$ 
	is a causal mismatch for $B$ at time $0$. 
	Therefore, a delay must start at $0$ for any $f\in \causalsolution(RN)$ satisfying $f_A(0)=f_B(0)=0$. 
	Over this delay, $f$ must be a solution of $\fix[\{B\}](RN)$.
	Let $f^1\in \fix[\{B\}](RN)$ be a function such that $f^1(0) = f^0(0)$.
	The function $f^1$ is uniquely determined as follows:
	$$ f^1_X = \left\{\begin{array}{ll}
		1 &\text{ if } X=r1\\
		0 &\text{ if } X=B\\
		exp(-t)&\text{ if } X=\rond{A}\\
		1-exp(-t) &\text{ if } X \in \{A,r_2,\rond{B}\}.\\
	\end{array}\right.
	$$	
	We compute now $\Absimage{f^1}$.
	First, for $t = 0$, we have:
	$$ \Absimage{f^1_X}(0)= \left\{\begin{array}{ll}
		(1,0) &\text{ if } X\in \{r_1,\rond{A}\}\\
		(0,0) &\text{ if } X=B\\
		(0,1) &\text{ if } X \in \{A,r_2,\rond{B}\}.\\
	\end{array}\right.
	$$ 
	Then, for any $t >0$, it holds that:
	$$ \Absimage{f^1_X}(t)= \left\{\begin{array}{ll}
		(1,0) &\text{ if } X\in \{r_1\}\\
		(1,-1)&\text{ if } X\in \{\rond{A}\}\\
		(0,0) &\text{ if } X=B\\
		(1,1) &\text{ if } X \in \{A,r_2,\rond{B}\}.\\
	\end{array}\right.
	$$ 
	We observe that, for any $t_1> 0$, $\Absimage{f^1}$ is constant on $]0,t_1]$. 
	Thus, for all $t_1> 0$, there is a function $f\in \causalsolution(RN)$ satisfying $f_A(0)=f_B(0)=0$ 
	with a delay $]0,t_1[$. We arbitrarily choose $t_1 = 0.5$. 
	Since $f$ is continuous, we have $f(0.5)=f^1(0.5)$. 
	Moreover, as $0.5$ does not belong to any delay of $f$, 
	there exists $f^2 \in \solrond(\ADE{RN})$ such that $f$ and $f^2$ are equal over $[0.5,t_2]$ 
	for some $t_2 \geq 0.5$. If a mismatch occurs after $t_1$ for $f^2$, the time at which it occurs is the maximal allowed value for $t_2$. 
	We denote the real $u = exp(-0.5) \approx 0.6$.
	The function $f^2$ is uniquely determined and satisfies:
	$$ f^2_X(t) = \left\{\begin{array}{ll}
		1 &\text{ if } X=r_1\\
		t + exp(-t) -u - 0.5 &\text{ if } X=B\\
		exp(-t)&\text{ if } X=\rond{A}\\
		1-exp(-t)&\text{ if } X \in \{A,r_2,\rond{B}\}.\\
	\end{array}\right.
	$$
	In order to show that $f^2$ is causal for all the species over $[0.5,+\infty[$, 
	we compute $\Absimage{f^2}$ over this interval.
	First, for $t = 0.5$, we have:
	$$ \Absimage{f^2_X}(0.5)= \left\{\begin{array}{ll}
		(1,0) &\text{ if } X =r_1\\
		(1,-1)&\text{ if } X=\rond{A}\\
		(0,1) &\text{ if } X=B\\
		(1,1) &\text{ if } X \in \{A,r_2,\rond{B}\}.\\
	\end{array}\right.
	$$ 
	Then, for any $t > 0.5$, we have:
	$$ \Absimage{f^2_X}(t)= \left\{\begin{array}{ll}
		(1,0) &\text{ if }  X =r_1\\
		(1,-1)&\text{ if } X=\rond{A}\\
		(1,1) &\text{ if } X \in \{A,B,r_2,\rond{B}\}.\\
	\end{array}\right.
	$$ 
	Now we can verify that $f^2$ is causal for $A$ and $B$. According to Equation~\ref{eq: causal}, theirs sign of differential and derivatives should be equal. 
	
	The sign of the differential of $X\in \{A,B\}$ is captured by the second component of $\Absimage{f^2_X}(t)$ and the sign of the derivative is
	the first component of $\Absimage{f^2_{\rond{X}}}(t)$. All these value are equal to $1$,
	therefore $f^2$ is causal for both species. 
	No delay is needed after $t=0.5$, therefore $f$ is equal to $f^2$ over $[0.5,+\infty[$.
	In summary, the function $f$ defined as follows, 
	$$
	f(t)= \left\{\begin{array}{ll} 
		f^0(t)  &\text{ if } t = 0 \\
		f^1(t)  &\text{ if } 0< t < 0.5\\
		f^2(t)  &\text{ if } t\geq 0.5 \\\end{array}\right.
	$$
	is in $\causalsolution(RN)$ and satisfies $f_A(0)=f_B(0)=0$.
	This function is plotted on the right part of the Figure~\ref{fig:causal}.
\end{example}
\subsection{Abstraction of causal continuous predictions}
We can now formally introduce the notion of causal next, 
which is constructed similarly to the temporal next, 
but based on causal predictions.

Delays are used to temporally fix some species in order to reintroduce causality. Therefore if $]t_1,t_2[$ is a delay, the natural immediate successor of the state at $t_1$ is the state at $t_2$. Hence, as stated in the following definition, abstract transitions on causal predictions are selected outside delays.
\begin{definition}[Causal Next]\label{def: cnext}
	Let $N=(S,R)$ be a reaction network. 
	The \emph{causal next relation of $N$ induced by $\homo$} is defined as 
	$$
	\rondcnextstep[\Sset](N)~= \bigcup_{f\in \causalsolution(N)} \ATrans{\homo}{\restr{f}{\R_+\setminus\cSeg}}.
	$$
	Similarly, the \emph{causal next relation over $N$ induced by $\Absfun$} is defined as 	
	$$
	\rondcnextstep[\slopeset](N) = \bigcup_{f\in \causalsolution(N)}
	\ATrans{\Absfun}{\restr{f}{\R_+\setminus\cSeg}}.
	$$	
\end{definition}
We point out that for a given reaction network $N$, $\rondcnextstep[\Sset](N)$ is a relation over states in $\Sset^{\labels (R)\cup\rond{S}\cup S}$.
On the other hand, $\rondcnextstep[\slopeset](N)$ is a relation over states $\slopeset^{\labels (R)\cup\rond{S}\cup S}$, with
$\slopeset = \Sset\times \Sset$. 

In figures, for the sake of readability, we project the states onto the signs of species and derivatives, that is  $\Sset^{\rond{S}\cup S}$. 
\begin{figure}
	\centering
	\includegraphics[width=5cm]{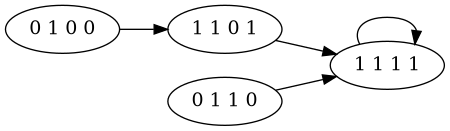}
	\caption{\label{fig:cnext}
		The graph induced by $\rondcnextstep[\slopeset](RN)$
		and projected on $\Sset^{\{A,\rond{A},B,\rond{B}\}}$ in this specific order. }
\end{figure}

\begin{example}
	We now compute the transitions obtained from the causal continuous prediction $f$ computed in the Example~\ref{ex: causal}.  By Definition~\ref{def: cnext}, these transitions belong to $\rondcnextstep[\slopeset](N)$.
	The first state is the abstraction of the initial condition: 
	$$\Absimage{f}(0)(X)= \left\{\begin{array}{ll}
		(1,0) &\text{ if } X=r1\\
		(0,0) &\text{ if } X=B\\
		(1,-1)&\text{ if } X=\rond{A}\\
		(0,1) &\text{ if } X \in \{A,r_2,\rond{B}\}.\\
	\end{array}\right.
	$$ 
	By projecting onto the signs of $A, \rond{A}, B, \rond{B}$, we obtain $(0,1,0,0)$. 
	Then, as $]0,0.5[$ is a delay, the immediate next step occurs at $t = 0.5$,
	and we have:
	$$\Absimage{f}(0.5)(X)= \left\{\begin{array}{ll}
		(1,0) &\text{ if } X=r1\\
		(0,1) &\text{ if } X=B\\
		(1,-1)&\text{ if } X=\rond{A}\\
		(1,1) &\text{ if } X \in \{A,r_2,\rond{B}\}.\\
	\end{array}\right.
	$$ 
	By projecting onto the signs of $A, \rond{A}, B, \rond{B}$ we obtain $(1,1,0,1)$. 
	Then the time step can be of any length, as long as no state is missed from the perspective of $\Absimage{\cdot}$.
	In this case, for any $t>0.5$, the abstraction will always be: 
	$$\Absimage{f}(t)(X)= \left\{\begin{array}{ll}
		(1,0) &\text{ if } X=r1\\
		(1,1) &\text{ if } X=B\\
		(1,-1)&\text{ if } X=\rond{A}\\
		(1,1) &\text{ if } X \in \{A,r_2,\rond{B}\}.\\
	\end{array}\right.
	$$ 
	By projecting again, we obtain $(1,1,1,1)$.
	Therefore, the trace abstracted from the function $f$, computed in Example~\ref{ex: causal} starts in 
	$(0,1,0,0)$, goes to $(1,1,0,1)$, and ends in $(1,1,1,1)$ (where it stays forever). Figure~\ref{fig:cnext} depicts the graph induced by $\rondcnextstep[\slopeset](RN)$ 
	that includes this path.
\end{example}
	\section{Abstract simulation based on the sign of the differential}\label{sec4}
 
 In this section, we introduce a new successor based on the differential sign abstraction and establish its soundness with respect to causal continuous predictions. 
 With this aim, we first define a new structure based on $\slopeset$ and show that the abstraction operator $\Absimage{\cdot}$ behaves similarly to a homomorphism. 
 Building on this perspective, we reinterpret continuity of real functions through the lens of this abstraction and derive the successor based on the differential sign abstraction. 
 Finally, we illustrate its efficiency compared to the sign successor (Definition \ref{def: sign-succ}) and prove that this successor captures all causal next transitions (Definition \ref{def: cnext}).

 \subsection{Differential sign structure}	
 We have shown that the sign structure is not sufficiently informative, 
 as it only expresses whether reactions are active. 
 Consequently, whenever a species is both produced and consumed by at least one active reaction, 
 the sign abstraction cannot predict whether the species is increasing, stable, or decreasing.
 However, for example, if the consumption of a species decreases while its production rises, 
 the derivative of the species must increase. 
 Similarly, if the reactions influencing a species remain stable, 
 its derivative remains stable.
 
 \label{page:diffstructure}
 To enrich the sign abstraction with these observations, 
 we combine it with the differential sign abstraction (Definition~\ref{def: diff}).
 For better visualization, we represent the set of signs associated with the differential as arrows: 
 $-1$ is denoted by $\downarrow$, $0$ by $\rightarrow$, and $1$ by $\uparrow$. 
 Therefore, from now on, $\slopeset$ will be equivalently seen as $\{-1,0,1\}\times\{\downarrow,\rightarrow,\uparrow\}$.
 \begin{definition}[Differential sign structure]\label{def: slope}
 	The \emph{differential sign structure} $\slope$ is the $\Symbs$-structure such that:  
 	\begin{itemize}
 		\item $\dom{\slope} = \slopeset$;
 		\item for any 
 		$k \in \Consts$, $k^{\slope} = (k^{\s},\rightarrow)$;
 		
 		\item for any operator $\odot$ in $\Symbs^{(n)}$, 
 		$$\odot^{\slope} = \{(\Absimage{a_1}(t),\dots, \Absimage{a_{n+1}}(t)) | t \in \R, (a_1 ,\dots, a_{n+1} ) \in \odot^{\Fun}\}.$$
 	\end{itemize}
 \end{definition}
 \begin{example}
 	\label{ex:x-1}
 	
 	The function $\fexpr[x-1]$ is positive for $x<1$ and increasing in the whole domain of definition. This can be represented in $\slope$
 	by the pairs: $(-1,\uparrow)$ when $x<1$, $(0,\uparrow)$ for $x=1$, and $(1,\uparrow)$ when $x>1$.
 	
 	Let us now consider the function $\fexpr[(x-1)^2]$. For any $x<1$, it is positive and decreasing. 
 	Indeed, 
 	$$
 	(-1,\uparrow) \mult^{\slope}_{set}(-1,\uparrow) = \{(1,\downarrow)\}.
 	$$
 	When $x=1$, the function $\fexpr[x-1]$ is zero but increasing.
 	For $\fexpr[(x-1)^2]$, we have
 	$$
 	(0,\uparrow) \mult^{\slope}_{set}(0,\uparrow) = \{ (0,\uparrow)\}.
 	$$
 	To conclude, for any $x>1$, $\fexpr[x-1]$ is positive and increasing, as well as $\fexpr[(x-1)^2]$.
 	Indeed, 
 	$$
 	(1,\uparrow) \mult^{\slope}_{set}(1,\uparrow) = \{ (1,\uparrow)\}.
 	$$
 \end{example}
 
 The full tables for arithmetic operators of $\slope$ are given in Appendix \ref{ap:delta9}.
 
 As $\slope$ is built from $\Fun$ through $\Absimage{\cdot}$, we naturally obtain the following lemma.
 \begin{lemma}\label{lemma:dt9homorphism}
 	Let $V$ be a finite set of variables and $\phi \in \formulas[\Symbs]{V}$ be a $\Symbs$-formula. 
 	For all $f \in \sol(\phi)$ and for all $t\in \DomF{f}$, 
 	the assignment $\alpha$ defined by $$\forall x \in V,~\alpha(x) = \Absimage{f_x}(t)$$ is a solution to $\phi$ in $\slope$, i.e. $\alpha \in \sol[\slope](\phi)$.
 \end{lemma}
 \begin{proof}
 	Let us consider a finite set of variables $V$, a function $f :V \to \dom{\Fun}$, and a time $t$ such that $t\in \DomF{f_x}$ for all $x\in V$.
 	We denote by $\alpha: V \to \slopeset$ the assignment such that:
 	$$\forall x \in V, \alpha(x) = \Absimage{f_x}(t).$$
 	We prove the statement by induction first on expressions and then on formulas. 
 	
 	For a given expression $e \in \Expr[\Symbs](V)$, let \(P(e)\) denote the property:
 	$$\{\Absimage{g}(t) \mid g\in \crochet{e}{f,\Fun}\} \subseteq \crochet{e}{\alpha,\slope}.$$
 	We prove that \(P(e)\) always holds by structural induction on the shape of $e$.
 	\begin{itemize}
 		\item If $e = k$ for $k \in \Consts$, we have $\crochet{k}{f,\Fun} = \{k^{\Fun}\}$,
 		which is a constant function. Therefore the sign of its differential is zero, i.e
 		$\Absimage{k^{\Fun}}(t) = (\homo(k),\rightarrow)$ and $\homo(k) = k^{\s}$.
 		According to the definition of $\slope$ (Definition~\ref{def: slope}), we have that 
 		$k^{\slope}= (k^{\s},\rightarrow)$.
 		Hence, the property holds.
 		
 		\item If $e = x $ with $x\in V$, according to the interpretation of expressions given in Section~\ref{para: expr}, it holds that 
 		$\crochet{x}{f,\Fun} = \{f_x\}$ and  $\crochet{x}{\alpha,\slope} = \{\alpha(x)\}$. Therefore, the property holds. 
 		
 		\item If $e = \odot(e_1)$ with $\odot \in \Symbs^{(1)}$ and $e_1 \in \Expr[\Symbs](V)$, according to the interpretation of expressions, 
 		given a $g \in \crochet{\odot(e_1)}{f,\Fun}$, it always exists a $g_1\in \crochet{e_1}{f,\Fun}$ such that $(g_1,g) \in \odot^{\Fun}$. 
 		By induction hypothesis, we have that $\Absimage{g_1}(t) \in \crochet{e_1}{\alpha,\slope}$. Moreover by definition of $\slope$, it follows that 
 		$(\Absimage{g_1}(t),\Absimage{g}(t)) \in \odot^{\slope}$. Therefore, we have that 
 		$\Absimage{g}(t) \in \crochet{e}{\alpha,\slope}$ and the property holds.
 		
 		\item If $e = e_1 \odot e_2$ with $\odot \in \Symbs^{(2)}$ and $e_1,e_2 \in \Expr[\Symbs](V)$, according to the interpretation of expressions, 
 		given a $g \in \crochet{e_1 \odot e_2}{f,\Fun}$, it exists a $g_1\in \crochet{e_1}{f,\Fun}$ and a $g_2\in \crochet{e_2}{f,\Fun}$ such that $(g_1,g_2,g) \in \odot^{\Fun}$.
 		By induction hypothesis, we have that $\Absimage{g_i}(t) \in \crochet{e_i}{\alpha,\slope}$, for $i \in\{1,2\}$. Moreover by definition of $\slope$,
 		it holds that $(\Absimage{g_1}(t),\Absimage{g_2}(t),\Absimage{g}(t)) \in \odot^{\slope}$. Therefore, it follows that 
 		$\Absimage{g}(t) \in \crochet{e}{\alpha,\slope}$ and the property holds.		 		 
 	\end{itemize}
 	Using the fact that \(P(e)\) holds for all expressions \(e\), we now prove by structural induction on the shape of formulas $\phi \in \formulas[\Symbs]{V}$ 
 	that the property holds:
 	$$f \in \sol(\phi) \Rightarrow \alpha\in \sol[\slope](\phi).$$
 	\begin{itemize}
 		\item If $\phi = (e_1\EQUAL e_2)$ with $e_1,e_2 \in \Expr[\Symbs](V)$, a function $f$ is a solution of $\phi$ whenever it exists a $ g \in \crochet{e_1}{f,\Fun}\cap  \crochet{e_2}{f,\Fun}$. By $P(e_1)$ and $P(e_2)$, we have that 
 		$\Absimage{g}(t) \in \crochet{e_1}{\alpha,\slope} \cap \crochet{e_2}{\alpha,\slope}$. 
 		Hence, the property holds.
 		\item If $\phi = \phi_1 \wedge \phi_2$ with $\phi_1,\phi_2 \in  \formulas[\Symbs]{V}$, a function $f$ is a solution of $\phi$ whenever  $f \in \sol(\phi_1)$ and $f \in \sol(\phi_2)$. By induction hypothesis, we have that  $\alpha\in \sol[\slope](\phi_1)$
 		and  $\alpha\in \sol[\slope](\phi_2)$.
 		Therefore, the property holds.	
 	\end{itemize}
 \end{proof}

 \subsection{Simulation graph for the differential sign abstraction}
 In Definitions~\ref{def: euler-succ} and \ref{def: sign-succ}, the notion of successor was introduced by using an approximation of the derivative. 
 However, as shown before, the derivative on its own fails to capture the changes of interest in the dynamics. 
 Here, we propose a new successor based on the differential sign abstraction which gives a more refined simulation graph, since it captures all the transitions thanks to the sign of the differential while respecting basic continuity constraints.  
 In particular, a trajectory should not transition directly from a positive value to a negative one without passing through zero, and vice versa.
 
 \medskip
 
 To start, we introduce a new notion of successor in order to model the possible abstract transitions of continuous functions.
 
 \label{page:continuoussucc}
 \begin{definition}[Continuous-compatible successor]\label{def:continuity}
 	The \emph{continuous-compatible successor} relation $\Continue$ is defined as the set of pairs $\bigl((s_1,\delta_1), (s_2,\delta_2)\bigr) \in \dom{\slope}\times \dom{\slope}$ such that:
 	\begin{enumerate}
 		\item\label{constr: passzero} if $s_1 \neq 0$, $s_2 \in \{s_1,0\}$;
 		\item\label{constr: diffrole}$s_2 \in s_1 +^{\s}_{\set} \delta_1$;
 		\item\label{constr: immediate} if $s_1 = 0$ and $\delta_1 \neq \rightarrow$, $\delta_2 = \delta_1$.
 	\end{enumerate}
 \end{definition}
 In this last definition, the first point corresponds to the notion of continuity: a continuous function cannot transition from a nonzero value to a value of the opposite sign without crossing zero. The second point follows directly from the fact that $\delta_1$ denotes the sign of the differential, and therefore determines the direction of change of the function. Finally, the third point models that, if the differential sign is nonzero, then there exists a time interval (possibly arbitrarily small) during which the function preserves its direction of change. In particular, if the value of the function is $0$, on a possibly arbitrarily small time interval, the value has changed while the differential sign remains unchanged.
 
 \smallskip
 
 We can express $\Continue$ as a function $\Continue_{set}: \dom{\slope} \to \mathcal{P}(\dom{\slope})$ such that, for every $v\in \dom{\slope}, \Continue_{set}(v) = \{v' \mid (v,v')\in \Continue\}$.
 The $\Continue_{set}$ function falls down to be: 
 $$\begin{array}{lll}
 	\Continue_{set}(-1,\downarrow) &= & \{(-1,\downarrow),(-1,\rightarrow),(-1,\uparrow)\}\\
 	\Continue_{set}(-1,\rightarrow) &= & \{(-1,\downarrow),(-1,\rightarrow),(-1,\uparrow)\}\\
 	\Continue_{set}(-1,\uparrow) &= & \{(-1,\downarrow),(-1,\rightarrow),(-1,\uparrow),(0,\downarrow),(0,\rightarrow),(0,\uparrow)\}\\
 	\Continue_{set}(0,\downarrow) &= & \{(-1,\downarrow)\}\\
 	\Continue_{set}(0,\rightarrow) &= & \{(0,\downarrow),(0,\rightarrow),(0,\uparrow)\}\\
 	\Continue_{set}(0,\uparrow) &= & \{(1,\uparrow)\}\\
 	\Continue_{set}(1,\downarrow) &= & \{(1,\downarrow),(1,\rightarrow),(1,\uparrow),(0,\downarrow),(0,\rightarrow),(0,\uparrow)\}\\
 	\Continue_{set}(1,\rightarrow) &= & \{(1,\downarrow),(1,\rightarrow),(1,\uparrow)\}\\
 	\Continue_{set}(1,\uparrow) &= & \{(1,\downarrow),(1,\rightarrow),(1,\uparrow)\}.\\
 \end{array}$$
 
 We now prove that the Continuous-compatible successor captures all and only the abstract transitions that are possible in continuous functions. 
 
 \begin{lemma}\label{lemma:ct} The set of continuous-compatible successors $\Continue$ is the union of differential sign transitions of functions of class $\mathcal{C}^1$ with a connected domain of definition, i.e.
 	
 	$$\Continue = \bigcup_{f \in  \mathcal{C}^1} \ATrans{\Absfun}{f}.$$
 	
 \end{lemma}
 \begin{proof}
 	$(\supseteq)$ Considering any function $f\in \mathcal{C}^1$ and any pair $((s_1,\delta_1),(s_2,\delta_2)) \in \ATrans{\Absfun}{f}$, we first show that $ \Continue \supseteq \ATrans{\Absfun}{f}$. 	By Definition \ref{def:seq}, it exists $t_1$ and $t_2$ such that $t_1<t_2$, $(s_1,\delta_1) = \Absimage{f}(t_1)$, $ (s_2,\delta_2) = \Absimage{f}(t_2)$, and $ \Absimage{f}$ is constant on $[t_1,t_2[$ or $]t_1,t_2]$. In order to have $((s_1,\delta_1), (s_2,\delta_2)) \in \Continue$, the three conditions of Definition \ref{def:continuity} must be satisfied.
 	
 	\textbf{Constraint \ref{constr: passzero}: if $s_1 \neq 0$ then $s_2 \in\{0,s_1\}$.} 
 	In this case, we need to prove that $s_2 \neq - s_1$. 
 	Let us suppose that $s_2 = - s_1$. In this case, by continuity of $f$, it exists a $t\in ]t_1,t_2[$, such that $f(t)=0$.
 	This is a contradiction with the fact that $ \Absimage{f}$ is constant over $[t_1,t_2[$ or $]t_1,t_2]$.
 	Therefore,  $s_2 \in\{0,s_1\}$ and $((s_1,\delta_1),(s_2,\delta_2))$ satisfies Condition~\ref{constr: passzero}.
 	
 	\textbf{Constraint \ref{constr: diffrole}: $s_2 \in s_1 +^{\s}_{\set} \delta_1$.} 
 	Since, it is trivially true that for any $t\in \R_+$,$$(f(t) - f(t_1)) + f(t_1) = f(t),$$
 	applying the homomorphism $\homo$, for $t\in \R_+$, we have:
 	\begin{equation}\label{eq:trivhomo}\homo(f(t)) \in  \homo(f(t_1)) +^{\s}_{set} \homo(f(t) - f(t_1)).
 	\end{equation}
 	We recall that $$\delta_1 = \lim_{\epsilon \to 0_+}  \Bigl( \homo \bigl(f(t_1+\epsilon)- f(t_1)\bigr)\Bigr).$$
 	As $\slopeset$ is a finite set and $\delta_1$ exists (since $\Absimage{f}(t_1)$ is defined), the limit is reached for some $\epsilon_1 >0$. Therefore, for any $\epsilon$ such that $0<\epsilon\leq \epsilon_1$, it holds that
 	\begin{equation}\label{eq:appdiff}\delta_1 = \homo (f(t_1+\epsilon) - f(t_1)).
 	\end{equation}
 	
 	Let us now consider two possibles cases:  $\Absimage{f}$ is constant on $]t_1,t_2]$ or $\Absimage{f}$ is constant on $[t_1,t_2[$.
 	
 	If $\Absimage{f} \text{ is constant on } ]t_1,t_2]$ then with an $\epsilon$ smaller than $t_2$, $\homo(f(t_1+\epsilon)) = s_2$. 
 	Therefore, by applying \eqref{eq:trivhomo} to $t =t_1+\epsilon$ and with this last observation we obtain $$s_2 \in  \homo(f(t_1)) +^{\s}_{set} \homo(f(t+\epsilon) - f(t_1)).$$ With \eqref{eq:appdiff} and the definition $s_1 =  \homo(f(t_1))$, we obtain $s_2 \in s_1 +^{\s}_{set} \delta_1$.
 	
 	If $\Absimage{f}$ is constant on $[t_1,t_2[$, we need to consider two possibilities.
 	If $f(t_1) = f(t_2)$, it holds that $\homo(f(t_1+\epsilon)) = s_2$ and we therefore apply the previous reasoning.
 	If $f(t_1)\neq f(t_2)$, by continuity and the constancy of $\Absimage{f}$ on $[t_1,t_2[$, we have that $f(t_1)$, $f(t_1+ \epsilon')$ and $f(t_2)$ are sorted by increasing or decreasing order, i.e.
 	$$\forall\epsilon' \in \,]0 ,t_2-t_1[, \,\,\,\,\,\,f(t_1)\lessgtr f(t_1 + \epsilon')\lessgtr f(t_2).$$ 
 	Hence, $\homo(f(t_2)-f(t_1)) = \homo((f(t_1+\epsilon) - f(t_1))$. 
 	Thus, applying \eqref{eq:appdiff}, we obtain that $\homo(f(t_2)-f(t_1)) = \delta_1$. Finally, according to \eqref{eq:trivhomo} with $t=t_2$, we obtain $s_2 \in s_1 +^{\s}_{set} \delta_1$.

 	\textbf{Constraint \ref{constr: immediate}: if $s_1 = 0$ and $\delta_1 \neq \rightarrow$ then $\delta_1 = \delta_2$.} 
 	According to \eqref{eq:appdiff}, we know that the abstraction of the differential is constant for $[t_1,t_1+\epsilon[$. 
 	However, the sign of $f$ will be different from $0$ immediately after $t_1$. 
 	Therefore, the abstracted successor state, for every time point in $[t_1,t_2[$, will have $s_2 \neq 0$ and $\delta_1=\delta_2$.
 	To conclude, we have that $((s_1,\delta_1),(s_2,\delta_2))$ satisfy Condition~\ref{constr: immediate}.
 	
 	\bigskip
 	
 	$(\subseteq)$ It remains to show that for all pair $((s_1,\delta_1),(s_2,\delta_2))\in \Continue$, there exists a $f \in \mathcal{C}^1$ such that $((s_1,\delta_1),(s_2,\delta_2))\in\ATrans{\Absfun}{f}$. We can separate these transitions into three categories.
 	
 	\textbf{Transitions $((s_1,\delta_1),(s_2,\delta_2))$ such that $\delta_1 = \delta_2$:} these transitions 
 	can be abstracted from affine functions. antonelli cristal inria lille

 	\textbf{Transitions such that the differential sign changes to its opposite (i.e., $\delta_1 = -\delta_2$):} these transitions can be obtained as abstractions of parabolas. Indeed, we have:
 	\begin{itemize}
 		\item $\fexpr[x^2]$ gives $\left((1,\downarrow), (0,\uparrow)\right)$ and the opposite $\fexpr[-x^2]$ gives $\left((-1,\uparrow), (0,\downarrow)\right)$;
 		\item $\fexpr[x^2+1]$ gives $\left((1,\downarrow), (1,\uparrow)\right)$ and $\fexpr[-x^2-1]$ gives $\left((-1,\uparrow), (-1,\downarrow)\right)$;
 		\item $\fexpr[-x^2+1]$ gives $\left((1,\uparrow), (1,\downarrow)\right)$ and $\fexpr[x^2-1]$ gives $\left((-1,\downarrow),(-1,\uparrow)\right)$.
 	\end{itemize}
 	
 	\textbf{Transitions where one of the two differential signs is zero}: here, we can
 	consider the following family of functions:		
 	$$ g_{a,b}(x) = \left\{\begin{array}{l}
 		a\mult x^2  + b\text{ for } x\leq 0 \\
 		b\text{ for } 0\leq x \leq 1 \\
 		a(x-1)^2 +b \text{ for } x > 1. \\
 	\end{array}\right.
 	$$
 	
 	All functions of this family are $\mathcal{C}^1$. 
 	The transitions from $\Absimage{g_{a,b}}(0)$ for any $a,b\in \{-1,0,1\}$ correspond to the set of all transitions $((s_1,\delta_1),(s_2,\rightarrow))\in \Continue$.
 	Similarly, the transitions from $\Absimage{g_{a,b}}(1)$or any $a,b\in \{-1,0,1\}$ correspond to the set of all transitions $((s_1,\rightarrow),(s_2,\delta_2))\in \Continue$.
 	
 	At this point, we have covered all the transitions in $\Continue$ with differential sign transitions of functions of class $\mathcal{C}^1$.
 	
 	To conclude, as we have both inclusions, we have $\Continue = \bigcup_{f \in  \mathcal{C}^1} \ATrans{\Absfun}{f}$.
 	
 \end{proof}
 
 Now, we will restrict this set of possible transitions according to the dynamics dictated by the differential equation $\constr{N}$.
 This idea is captured by the differential sign successors. 
 
 \begin{definition}[Differential sign successor]\label{def:diff-succ}
 	Let $N=(S,R)$ be a reaction network.
 	The \emph{differential sign successor} relation $\Succ[\slope](N)$ is defined as the set of pairs $(\alpha,\beta)$, with $\alpha,\beta \in \slopeset^{ S \cup\rond{S}\cup \labels (R)}$, satisfying the following three conditions:
 	\begin{enumerate}
 		\item \textbf{[Induced dynamics]}\label{constr: succdiff} $\alpha,\beta \in \sol[\slope](\constr{N}); $
 		
 		\item \textbf{[Continuity]} \label{constr: succcont}for all $ x \in S\cup \labels (R)\cup \rond{S}$,
 		$ (\alpha(x),\beta(x)) \in \Continue;$
 		\item  \textbf{[Causality]} \label{constr: succcausal} for all $x \in S$,
 		$\pi_2(\alpha(x)) = \pi_1(\alpha(\rond{x})) \text{ and }\pi_2(\beta(x)) = \pi_1(\beta(\rond{x})).$
 	\end{enumerate}
 	
 \end{definition}
 
 \begin{example}\label{ex:difsignsucc}
 	Let us apply Definition \ref{def:diff-succ} over the three reaction networks in Figure~\ref{fig:threeRN}.
 	Starting from $\Nsimpl$ and considering any initial state $\asimpl\in \slopeset^{\{A,\rond{A},r1,r2\}}$ such that $\asimpl(A) = (1,\rightarrow)$, we aim to compute the set of differential sign successors. 
 	Any such state $\asimpl$ may admit a successor if and only if it satisfies the following two conditions:
 	\begin{enumerate}[label=\Roman*]
 		\item $\asimpl \in \sol[\slope](\constr{\Nsimpl}),$ (from condition \ref{constr: succdiff}), \label{alphacond1}
 		\item $\forall x \in S,~\pi_2(\asimpl(x)) = \pi_1(\asimpl(\rond{x}))$ (from condition \ref{constr: succcausal}). \label{alphacond2}
 	\end{enumerate}
 	Recall that $\constr{\Nsimpl}$ is defined as follows :
 	$$\constr{\Nsimpl} := r1\EQUAL 1  \land r2  \EQUAL A \land \rond{A} \EQUAL r1 - r2.$$
 	From \ref{alphacond1}, we determine that $\asimpl(r1) = \asimpl(r2)= (1,\rightarrow)$.
 	At this point, we cannot uniquely determine $\asimpl(\rond{A})$:
 	$$\asimpl(\rond{A})  \in (1,\rightarrow) -^{\slope}_{set} (1,\rightarrow) = \Sset\times\{\rightarrow\}.$$	
 	Using \ref{alphacond2}, we can deduce $\pi_1(\asimpl(\rond{A})  ) = 0$, and conclude that $\asimpl(\rond{A})=(0,\rightarrow)$.
 	
 	Now that we have computed the complete initial state 
 	$$\asimpl := [A=(1,\rightarrow); \rond{A} =(0,\rightarrow); r1 = (1,\rightarrow); r2 = (1,\rightarrow) ], $$
 	let us compute a possible differential sign successor $\asimpl'$ of $\asimpl$.
 	We recall that $ \Continue_{set} (1,\rightarrow) = \{1\}\times\Sset $ and $ \Continue_{set} (0,\rightarrow) = \{0\}\times\Sset. $
 	With Condition \ref{constr: succcausal}, we obtain $\asimpl'(A) = (1,\rightarrow)$. 
 	Since a differential sign successor is a solution of $\constr{\Nsimpl}$, we obtain that
 	$\asimpl'=\asimpl$ and therefore $\asimpl$ is the unique differential sign successor of itself. 
 	
 	We now proceed to apply Definition \ref{def:diff-succ} on $\NAct$ considering an initial state $\aAct\in \slopeset^{\{A,\rond{A};\Act,\rond{\Act},\rAct,r1,r2\}}$ such that $\aAct(A) = (1,\rightarrow)$, and $\aAct(\Act) = (1,\uparrow)$, we aim to compute its differential sign successors.
 	In order to admit a successor, $\aAct$ must be a solution of $\constr{\NAct}$, defined as follows:
 	$$\constr{\NAct} := \rAct \EQUAL 1\land r1\EQUAL \Act  \land r2  \EQUAL A \land \rond{\Act}\EQUAL \rAct \land \rond{A} \EQUAL r1 - r2.$$
 	From this, we determine that $\aAct(\rAct) = \aAct(\rond{\Act})=\aAct(r2) = (1,\rightarrow)$, and $\aAct(r1)=\aAct(\Act) = (1,\uparrow)$. 
 	At this point, we cannot uniquely determine $\aAct(\rond{A})$ since we have
 	$$\aAct(\rond{A})  \in (1,\uparrow) -^{\slope}_{set} (1,\rightarrow) = \Sset\times\{\uparrow\}.$$	
 	Using Condition \ref{constr: succcont}, we can deduce $\pi_1(\aAct(\rond{A}) ) = 0$, and conclude that $\aAct(\rond{A})=(0,\uparrow)$.
 	Now that we have computed the complete initial state 
 	$$\aAct := [A=(1,\rightarrow); \rond{A} =(0,\uparrow); \Act = (1,\uparrow); \rond{\Act} = (1,\rightarrow); \rAct = (1,\rightarrow);r1 = (1,\uparrow); r2 = (1,\rightarrow) ], $$	
 	let us compute a possible differential sign successor $\aAct'$ of $\aAct$.
 	From Condition~ \ref{constr: succcont}, we have that the sign of each variable is $1$, therefore by Condition \ref{constr: succcausal}, we can determine that $\aAct(A) = \aAct(\Act) = (1,\uparrow)$.	
 	Since a differential sign successor is a solution of $\constr{\NAct}$, we obtain that	
 	$$\aAct' := [A=(1,\uparrow); \rond{A} =(1,\uparrow); \Act = (1,\uparrow); \rond{\Act} = (1,\rightarrow); \rAct = (1,\rightarrow);r1 = (1,\uparrow); r2 = (1,\uparrow) ]$$
 	is the only differential sign successor of $\aAct$. 
 	
 	The difference between $\NAct$ and $\NInh$ is the kinetics of $r1$. We consider an initial state $\aInh\in \slopeset^{\{A,\rond{A};\Inh,\rond{\Inh},\rInh,r1,r2\}}$, such that $\aInh(A) = (1,\rightarrow)$, and $\aInh(\Act) = (1,\uparrow)$, to compute the differential sign successors for $\NInh$.
 	
 	In $\constr{\NInh}$, we have $r1\EQUAL \frac{1}{1+\Inh}$ and since
 	$$ (1,\rightarrow) /^{\slope}_{set} \left(  (1,\rightarrow) +^{\slope}_{\set} (1,\uparrow)\right) = (1,\rightarrow) /^{\slope}_{set} (1,\uparrow) = \{(1,\downarrow)\},$$
 	we have $\aInh(r1) = (1,\downarrow)$. By applying the same reasoning on each variable and taking Condition \ref{constr: succcausal} into account, we obtain
 	$$\aInh := [A=(1,\rightarrow); \rond{A} =(0,\downarrow); \Inh = (1,\uparrow); \rond{\Inh} = (1,\rightarrow); \rInh = (1,\rightarrow);r1 = (1,\uparrow); r2 = (1,\rightarrow) ]. $$
 	The only differential sign successor of $\aInh$ is: 	
 	$$\aInh':= [A=(1,\downarrow); \rond{A} =(-1,\downarrow); \Inh = (1,\uparrow); \rond{\Inh} = (1,\rightarrow); \rInh = (1,\rightarrow);r1 = (1,\downarrow); r2 = (1,\downarrow) ]. $$		
 \end{example} 
 
 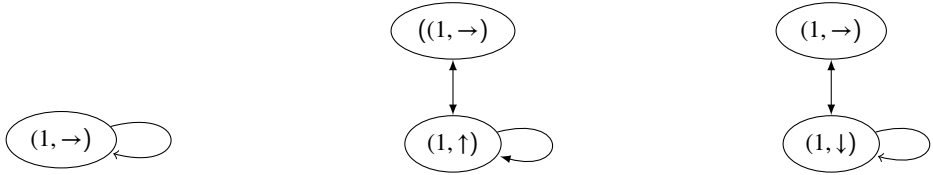
\begin{figure}
 	\centering
 	\begin{subfigure}{0.30\textwidth}
 		\centering
 		\begin{tikzpicture}[scale=1]
 			\node[ellipse, draw] (A) at (-2.5,1.5) {$(1,\rightarrow$)};
 			\path[->] (A) edge[loop right] ();
 		\end{tikzpicture}
 		\caption{\label{fig:diffsimpl}Transition graph from $\asimpl$ for $\Nsimpl$.}
 	\end{subfigure}
 	\begin{subfigure}{0.30\textwidth}
 		\centering
 		\begin{tikzpicture}[scale=1]
 			\node[ellipse, draw] (D) at (1.5,1.5) {($(1,\rightarrow$)};	
 			\node[ellipse, draw] (E) at (1.5,0) {$(1,\uparrow$)};
 			\draw[->,>=latex] (E) edge[loop right] ();
 			\draw[<->,>=latex] (E) to (D);	
 		\end{tikzpicture}
 		\caption{\label{fig:diffAct}Transition graph from $\aAct$ for $\NAct$.}	
 	\end{subfigure}
 	\begin{subfigure}{0.30\textwidth}
 		\centering
 		\begin{tikzpicture}[scale=1]
 			\node[ellipse, draw] (C) at (5,0) {$(1,\downarrow$)};
 			\node[ellipse, draw] (F) at (5,1.5) {$(1,\rightarrow$)};
 			\draw[<->,>=latex] (C) to (F);
 			\path[->] (C) edge[loop right] ();		
 		\end{tikzpicture}
 		\caption{\label{fig:diffInh}Transition graph from $\aInh$ for $\NInh$.}	
 	\end{subfigure}
 	\caption{\label{fig:difsignsharp}The projections on $A$ of the graphs obtained according to the differential sign successor relation of reaction network in Figure~\ref{fig:threeRN} starting from the initial states of Example \ref{ex:difsignsucc}.}
 \end{figure}
 In Figure \ref{fig:signtoocoarse}, the graph obtained according to the sign successor relation representing the possible trajectories from the initial state $(1,0)$ contains $11$ edges. 
 According to the differential sign successor relation, the number of possible successors is reduced to only one (see Figure \ref{fig:diffsimpl}). 
 This reduction in the number of successors is correct since the initial state should be stable. 
 Moreover, the differential sign successor faithfully captures the impact of activators and inhibitors, whereas the sign successor fails to capture their influence in the network. 
 Figure~\ref{fig:graphfinal} shows the graph obtained according to the differential sign successor that, in addition to properly treating the evolution of the system when some species are absent, it clearly represents the possible dynamics when both species are present: either with the oscillation or the stability of populations.
 Considering the number of edges in the graphs after projection onto the signs of species and their derivatives, $\Succ[\s](LV)$ turns out to have $92$ edges while $\Succ[\slope](LV)$ presents $29$ edges, and $\rondcnextstep[\slopeset](LV)$ contains $17$ edges. 
 This significant reduction highlights the increased precision of the differential sign successor with respect to sign successor.

 \begin{figure}
 	\centering
 	\includegraphics[width=10cm]{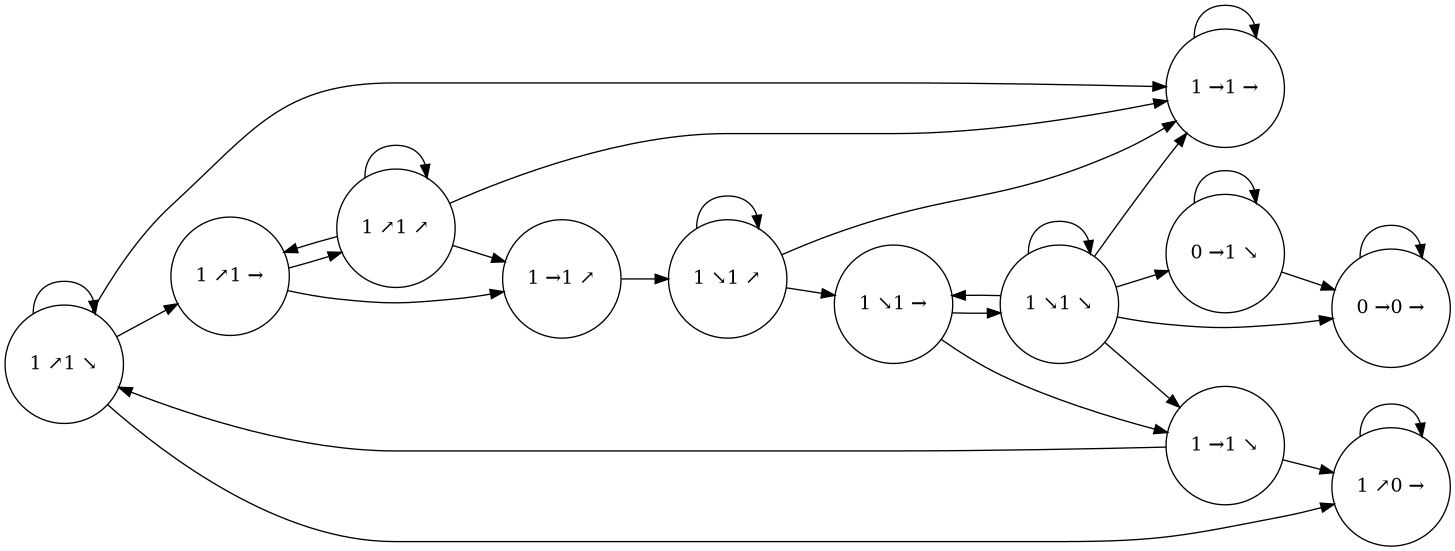}
 	\caption{\label{fig:graphfinal}The transition graph induced by the differential sign successor for the $LV$ network, projected on values of species $Ra,\rond{Ra},Fo,\rond{Fo}$ (in this order). Self-arrows indicate that at least one of the states projected on the vertex is a successor of itself.}
 \end{figure}

 We now show that given a reaction network, its differential sign successor is a sound abstraction of the causal continuous semantics.
 \begin{theorem}[Soundness]\label{th:soudness}
 	For any reaction network $N$, its differential sign successor includes its causal next, i.e. 
 	$$
 	\rondcnextstep[\slopeset](N) \subseteq \Succ(N).
 	$$
 \end{theorem}
 \begin{proof}
 	
 	We consider a reaction network $N=(S,R)$ and a pair of states $(\alpha,\beta) \in \rondcnextstep[\slopeset](N)$. 
 	According to Definition~\ref{def: cnext}, it exists a function $f \in \causalsolution(N)$, and two time points $t_1,t_2\in \DomF{f}\setminus \cSeg$  such that  $\Absimage{f}(t_1) = \alpha$ and $\Absimage{f}(t_2) = \beta$. 
 	We prove that $(\alpha,\beta)$ satisfies the three constraints of Definition~\ref{def:diff-succ}.
 	\begin{enumerate}
 		\item  \textbf{Induced dynamics: }$\alpha,\beta \in \sol[\slope](\constr{N})$.\\
 		The $\dot{\Symbs}$-structure $\Fun$ has the particularity to interpret each of its operators as a partial function.
 		This has the consequence that the interpretation of an expression in this structure is at most a singleton. 
 		Hence, any function $f$ is a solution of $e_1\EQUAL e_2$  if and only if the interpretations of $e_1$ and $e_2$ are equal to the same singleton. 
 		Therefore, with the structure $\Fun$, the operator $\EQUAL$ is interpreted as the standard algebraic equality so that we can substitute expressions in formulas. 
 		Let us recall that the formula $\solrond(\ADE{N})$ is defined as follows: 
 		$$ \bigwedge_{r\in \labels(R)} r \EQUAL e_r ~\bigwedge_{x\in S} \left(\dot{x} \EQUAL \exprdex \land \rond{x} \EQUAL \dot{x} \right)$$
 		where $e_x =  \sum_{r\in \labels (R)} \Stcoef{r}{x}\mult r$.
 		If we substitute $\dot{x}$ by $\rond{x}$ in formula $\solrond(\ADE{N}))$, we obtain
 		$$\bigwedge_{r\in \labels(R)} r \EQUAL e_r ~\bigwedge_{x\in S} \rond{x} \EQUAL e_x$$
 		which falls down to be $\constr{N}$.
 		This implies that $ \solrond(\ADE{N})) \subseteq \sol(\constr{N})$.
 		Using the same idea, we have that $ \sol(\fix(N))\subseteq \sol(\constr{N})$ for all $S_1\subseteq S$.
 		As $\causalsolution(N)$ is a set of functions constructed piecewise from these solutions, and $\constr{N}$ is a $\Symbs$-formula, we have $\causalsolution(N) \subseteq \sol(\constr{N})$. Therefore, we have $f\in \sol(\constr{N})$ and applying Lemma \ref{lemma:dt9homorphism}, 
 		we obtain that, for all $t\in \DomF{f}$, $\Absimage{f}(t)\in \sol[\slope](\constr{N})$.		
 		\item  \textbf{Continuity: }$\forall x \in S\cup\rond{S}\cup \labels (R), (\alpha(x),\beta(x)) \in \Continue$.\\
 		As  $(\alpha,\beta)\in\rondcnextstep[\slopeset](N)$, we have $(\alpha,\beta)\in 	\ATrans{\Absfun}{\restr{f}{\R_+\setminus\cSeg}}$.
 		By Constraint~\ref{constr: continuity}, each component $f_x$ of $f$ is continuous on a connected domain. According to Definition~\ref{def:seq}, we have two possible cases: either $\restr{f}{\R_+\setminus\cSeg}$ is defined on $[t_1,t_2]$ or it is not defined on $]t_1,t_2[$.
 		\begin{itemize}
 			\item If $[t_1,t_2]\subset \DomF{f} \setminus\cSeg$, then we trivially have that  $(\alpha,\beta)\in 	\ATrans{\Absfun}{f}$. Moreover, for each variable $x$, we have $(\alpha(x),\beta(x))\in 	\ATrans{\Absfun}{f_x}$.
 			\item Otherwise, $]t_1,t_2[$ is a delay (Definition ~\ref{def: delays}) of $f$. By Constraint~\ref{constr: small enough}, for each variable $x$, $\Absimage{f_x}$ is constant on either $[t_1,t_2[$ or $]t_1,t_2]$, therefore we have $(\alpha(x),\beta(x)) \in  \ATrans{\Absfun}{f_x}$.
 		\end{itemize}
 		
 		To conclude, in both cases we have that $(\alpha(x),\beta(x)) \in  \ATrans{\Absfun}{f_x}$ and $f_x$ is 
 		continuous on a connected domain therefore we can apply Lemma~\ref{lemma:ct} and obtain that $(\alpha(x),\beta(x)) \in \Continue$.
 		\item \textbf{Causality: } for all $x \in S$, $\pi_2(\alpha(x)) = \pi_1(\alpha(\rond{x})) \text{ and }\pi_2(\beta(x)) = \pi_1(\beta(\rond{x})).$\\
 		According to the definition of $\Absimage{\cdot}$, we have 
 		$$ \begin{array}{ll}
 			\pi_1(\alpha(\rond{x})) = \homo \circ f_{\rond{x}} (t_1),~&~ \pi_2(\alpha(x)) = \Absdiff{\homo}{f_x}(t_1),
 			\\
 			\pi_1(\beta(\rond{x})) = \homo \circ f_{\rond{x}} (t_2),~&~ \pi_2(\beta(x)) = \Absdiff{\homo}{f_x}(t_2).
 		\end{array}
 		$$ 
 		
 		From Constraint~\ref{constr: causal}, we have that $f$ is causal (Definition~\ref{def:match}) which implies that for each $x\in S$ 
 		and $t\in \DomF{f}$, we have $$ \Absdiff{h_\s}{f_x}(t) = \homo \circ \Dplus f(t).$$
 		It now remains to prove that for all $t\in \DomF{f}\setminus \cSeg$ (and a fortiori for $t_1$ and $t_2$) we have: 
 		$$ f_{\rond{x}} (t) = \Dplus f(t).$$
 		There are two possible cases, whether or not it exists $t' >t$ such that $]t,t'[$ is a delay.
 		If $t$ is not a left boundary of a delay, then by Constraint \ref{constr: real ADE}, there exist an open set $\mathcal{O}$ such that $t \in \mathcal{O}$ and a function $f'\in \solrond(\ADE{N})$ such that $\restr{f}{\mathcal{O}} = \restr{f'}{\mathcal{O}}$. Therefore, we have $\Dplus f_x (t) = \Dplus f'_x (t)$. 
 		Since $f'$ must satisfy $\dot{x} \EQUAL \rond{x}$ (by the fact that $f'\in \solrond(\ADE{N})$), we obtain $\Dplus f_x (t) = f_{\rond{x}}(t)$. 
 		On the other hand, if there exists a $t'$ such that $]t,t'[$ is a delay of $f$ for some subset of species $S_1$ then we have two cases:
 		\begin{itemize} 
 			\item if $x\notin S_1$, since $f$ is continuous and must satisfy $\dot{x} \EQUAL \rond{x}$ on $]t,t'[$ (by the fact that $\restr{f}{]t,t'[}\in \FIX{S_1}{N}$), we obtain $\Dplus f_x (t) = f_{\rond{x}}(t)$. 
 			
 			\item if $x\in S_1$, then $f_{\rond{x}}(t) = 0$ (Constraint~\ref{constr: small enough}) and the formula $\fix(N)$ contains $\dot{x}\EQUAL 0$. 
 		\end{itemize}
 		By substitution in the causality equation, we obtain  
 		$\pi_2(\alpha(x)) = \pi_1(\alpha(\rond{x}))$ and $\pi_2(\beta(x)) = \pi_1(\beta(\rond{x})).$
 		
 	\end{enumerate}
 	To conclude, $(\alpha,\beta)$ satisfies the three constraints to be in $\Succ(N)$ (Definition \ref{def:diff-succ}).
 	
 \end{proof}
	\section{Conclusion}

In this paper, the goal has been to tackle the challenge of defining a new qualitative abstraction for reaction networks which captures the continuous dynamics taking into account causality and distinguishing activations and inhibitions in reactions.
To achieve this, we first noted the non-causal nature of continuous semantics and observed that it arises from a discrepancy between the sign of the derivative and that of the differential. 
These mismatches stem from the fact that the continuous semantics is obtained as a limit in which the time step tends to zero, thereby intuitively merging the next state with the current one. 
Therefore, we reintroduced a delay between a state and its successor when necessary, which results in the definition of the causal continuous semantics. 

Later, we introduced a new structure, called differential sign structure, abstracting the sign of both the value and the differential of all the variables of the system.
Using this structure, we defined the differential sign successor and showed that it is a sound abstraction of the causal continuous semantics.
To best of our knowledge, we indeed obtained the first abstraction of continuous semantics providing exploitable information about the differential behaviour of the variables, allowing us to distinguish between activators and inhibitors (contrarily to the sign semantics of reaction Networks).
As shown in examples, thanks to our new method we obtained appropriate Boolean transition graphs that refine those provided by the previous approach.

Future works include possible improvements of our approach (for example being able to exploit the difference between point-wise and segment-wise transitions) and the refinement of existing techniques for reachability analysis between steady states \cite{allart_computing_2019,niehren_predicting_2016} by combining them with the differential sign successor.

Another interesting direction is to investigate the applicability of our approach to partial reaction networks \cite{niehren_predicting_2016} which represent the most extreme case of missing information, where not only exact information about concentrations and kinetics parameters is not available, but also the kinetic expressions of the reactions are partially unknown. 

In addition, for the analysis of systems where faithfulness to the standard semantics based on ODEs is required, we aim to study the result of the application of the differential sign abstraction directly over ODEs and its relation with the transition graph obtainable thanks to our approach.

	\section*{Acknowledgements}
	This was supported by the French Agence Nationale pour la Recherche (ANR) in the scope of the project “REBON” (grant number ANR-23-CE45-0008).
	The authors would also like to thank Dr. Loïc Paulevé and Dr. Patrick Baillot for their scientific discussions regarding this work.

	\bibliographystyle{model1-num-names}
	\bibliography{refCMSB2025.bib}
	
	\newpage
	\appendix 
\section{Appendix 1}
\subsection{Sign interpretation of arithmetic operators}

Here we express explicitly the sign interpretation of the operators $\{+,-,*,/ \}$. 

Let us start with $+^\s $:
$$
\begin{array}{lllllllllll}
	1 +^{\s} 1& = &\{1\},&~~~&~1 +^{\s} 0& = &\{1\},&~~~&~1 +^{\s} -1& = &\{1,0,-1\},\\
	0 +^{\s} 0 & = &\{0\},&~~~&~0 +^{\s} -1 & = &\{-1\},&~~~&-1+^{\s} -1 & = &\{-1\},\\
	~& & &\forall x, y\in dom(\s),	&~x +^{\s} y & = & y +^{\s} x.& & & \\
\end{array}
$$

We now express $-^\s $:
$$
\begin{array}{lllllllllll}
	1 -^{\s} 1& = &\{1,0,-1\};&~~~~&~0 -^{\s} 1 & = &\{-1\};&~~~~&-1 -^{\s} 1 & = &\{-1\};\\
	1 -^{\s} 0& = &\{1\};&~~~~&~0 -^{\s} 0 & = &\{0\};&~~~~&-1 -^{\s} 0 & = &\{-1\};\\
	1 -^{\s} (-1)& = &\{1\};&~~~~&~0 -^{\s}(-1)& = &\{1\};&~~~~&-1 -^{\s} (-1) & = &\{1,0,-1\}.\\
\end{array}
$$

The sign interpretation of multiplication $*^\s $ is as follows:

$$
\begin{array}{lllllllllll}
	1 *^{\s} 1& = &\{1\};&~~~&~1 *^{\s} 0& = &\{0\};&~~~&~1 *^{\s} -1& = &\{-1\};\\
	0 *^{\s} 0 & = &\{0\};&~~~&~0 *^{\s} -1 & = &\{0\};&~~~&-1*^{\s} -1 & = &\{1\};\\
	~& & &\forall x, y\in dom(\s),	&~x *^{\s} y & = & y *^{\s} x.& & & \\
\end{array}
$$
Finally, $/^\s$ correspond to the following table:
$$
\begin{array}{lllllllllll}
	1 /^{\s} 1& = &\{1\};&~~~~&~0 /^{\s} 1 & = &\{0\};&~~~~&-1 /^{\s} 1 & = &\{-1\};\\
	1 /^{\s} 0& = &\emptyset;&~~~~&~0 -^{\s} 0 & = &\emptyset;&~~~~&-1 -^{\s} 0 & = &\emptyset;\\
	1 /^{\s} (-1)& = &\{1\};&~~~~&~0 /^{\s}(-1)& = &\{0\};&~~~~&-1 /^{\s} (-1) & = &\{1\}.\\
\end{array}
$$

\subsection{Differential sign interpretation of arithmetic operators}\label{ap:delta9}
We now give the interpretation for the differential sign structure.
The interpretation of addition $+^{\slope} $ is given as follows:
$$(v_1,\delta_1) +^{\slope} (v_2,\delta_2) = \{(v,\delta)\mid v \in v1 +^\s v2 \text{ and } \delta \in \delta_1 +^\s \delta_2\}$$
The interpretation of multiplication $\mult^{\slope} $ is commutative and as follows:
\begin{align*}
	(0,\rightarrow) \mult^{\slope} x  = \{(0,\rightarrow)\}, ~~~& (0,\uparrow) \mult^{\slope} (0,\uparrow)  = \{(0,\uparrow)\}, ~~~ (0,\downarrow) \mult^{\slope} (0,\downarrow) = \{(0,\uparrow)\},\\
	&(0,\uparrow) \mult^{\slope} (0,\downarrow) = \{(0,\downarrow)\}.
\end{align*}
For any other pattern, we have
$$
(v_1,\delta_1)\mult^{\slope} (v_2,\delta_2) =\{(v,\delta) \mid v \in v_1 \mult^\s v_2; \delta \in (v_1 \mult^\s \delta_2) +^\s (\delta_1 \mult^\s v_2)\}.
$$
The interpretation of the minus operator corresponds exactly to the interpretation of $x + ((-1) \mult y)$, 
i.e. 
$$\forall x,y \in \Sset_d,~~ x-^{\slope}y = \crochet{v_x + ((-1) \mult v_y)}{[v_x/x;v_y/y],\slope} .$$
To express the last operator $/^{\slope}$, we will first express $1^{\slope}/^{\slope} x$ for all $x$.
$$
\begin{array}{lllllllllll}
	(1,\rightarrow)/^{\slope}_{set} (1, \uparrow) &=&\{(1,\downarrow)\};&~~&(1,\rightarrow)/^{\slope}_{set}
	(1, \rightarrow) &=&\{(1,\rightarrow)\};&~~&	(1,\rightarrow)/^{\slope}_{set} (1, \downarrow) &=&\{(1,\uparrow)\}; \\
	(1,\rightarrow)/^{\slope}_{set} (-1, \uparrow) &=&\{(-1,\downarrow)\};&~~&	(1,\rightarrow)/^{\slope}_{set} 
	(-1, \rightarrow) &=&\{(-1,\rightarrow)\};&~~&	(1,\rightarrow)/^{\slope}_{set}(- 1, \downarrow) &=&\{(-1,\uparrow)\}.\\
\end{array}			
$$
The interpretation of division $x /^{\slope}y$ is the result of $x \mult (1 / y)$. 
\subsection{Summary of notation}
\begin{tabular}{ l r l }
	Symbol &Reference & Signification \\
	\hline
	$\frac{ d x}{d t}$&Example \ref{ex:intro} &The derivative of $x$\\
	\hline
	$\Sigma$ &Definition \ref{def: signature}& A ranked signature \\
	\hline
	$\Consts$ &Definition \ref{def: signature}& A set of constants in a given ranked signature $\Sigma$ \\
	\hline 
	$\Sigma^{(n)}$ &Definition \ref{def: signature}& The set of n-ary operators of the ranked signature $\Sigma$ \\
	\hline 
	$\odot$  &Definition \ref{def: structure} & An operator \\
	\hline		
	
	$k$ &Definition \ref{def: structure} & A constant in $\Consts$ \\
	\hline 
	$\struct$ &Definition \ref{def: structure}& A relational structure\\
	\hline
	$\dom{\cdot}$ &Definition \ref{def: structure}& The domain of a relational structure\\
	\hline 
	$\Sigma_{\mathrm{arith}}$ &Page \pageref{def:Rarith} & The ranked signature composed of arithmetic operators \\
	\hline
	$\Rarith$ &Definition \ref{def:Rarith} & The structure of real numbers over the signature $\Sigma_{\mathrm{arith}}$\\
	\hline 
	$\Symbs$ &Definition \ref{def: actualsignature}& A ranked signature containing $\Sigma_{\mathrm{arith}}$\\ 
	\hline
	$\Romega$ &Definition \ref{def:Romega}& The structure of the real numbers over the signature $\Symbs$\\
	\hline
	$\R$ & Page \pageref{page:fun}& The set of real numbers and by abuse of notation $\Romega$\\
	\hline 
	$\dot{x}$ & Page \pageref{page:fun}& The unary derivative operator applied on $x$\\
	\hline
	$\dot{\Symbs}$ & Page \pageref{page:fun}& The signature obtained from $\Symbs$ by adding the unary operator $(\dot{~})$\\
	\hline 
	$\Fun$ &Definition \ref{def:fun}& The structure of functions over $\dot{\Symbs}$\\
	\hline 
	$\DomF{f}$ &  Page \pageref{page:fun}& The domain of the function $f$\\
	\hline
	$\idd$ & Example \ref{ex:fun}& The identity function \\
	\hline
	$\Sset$ &Page  \pageref{page:signset} & The set of signs $\{-1,0,1\}$\\
	\hline
	$\homo$ &Page  \pageref{page:signset} & The homomorphism transforming real values into their sign\\
	\hline
	$\s$ &Definition \ref{def:signstruct}& The structure of signs over $\Symbs$\\
	\hline
	$\mathcal{V}$&Page \pageref{para: expr} & The set of all variables\\
	\hline
	$V$ &Page \pageref{para: expr} & A set of variables, subset of $\mathcal{V}$\\
	\hline
	$\Expr[\Sigma](V)$ & Page \pageref{para: expr}& The set of expression over the signature $\Sigma$ and the set of variables $V$\\
	\hline
	$e$ & Page \pageref{para: expr}& An arbitrary expression \\
	\hline 
	$\alpha$ & Page \pageref{para: expr} & An assignment of variables, which can also be considered as a state of the system \\ 
	\hline
	$\crochet{e}{\alpha,\struct}$ &Page \pageref{para: expr} & The interpretation of an expression $e$ in $\struct$ with the assignment $\alpha$ \\
	\hline
	$\fexpr$ &Page \pageref{para: expr} & The function over $\R$ corresponding to the expression $e$ \\ 
	\hline
	$\formulas[\Sigma]{V}$ & Page \pageref{para: formulas}& The set of formulas with operators and constants in $\Sigma$ and variables in $V$\\
	\hline 
	$\phi,\psi$ &Page \pageref{para: formulas} &A formula \\
	\hline 
	$\EQUAL$ &Page \pageref{para: formulas} &The equality symbol insides formulas\\ 
	\hline
	$\fv(\cdot)$ &Page \pageref{para: formulas} &The set of free variables in the formula\\
	\hline
	$\sol[\struct](\phi)$ &Page \pageref{para: formulas} &The set of assignments satisfying $\phi$ in the structure $\struct$ \\
	\hline 
	$S$ & Page~\pageref{page:rn}&The set of species\\
	\hline
	$r$ & Definition~\ref{def:reaction}&A variable representing a reaction\\
	\hline
	$\Reac_r$ & Definition~\ref{def:reaction}&The multiset of reactants\\
	\hline 
	$\Pro_r$ & Definition~\ref{def:reaction}&The multiset of products\\
	\hline  
	$N$ &Definition~\ref{def:rn} &A reaction network\\
	\hline
	$R$ &Definition~\ref{def:rn} &The set of reactions of a network\\
	\hline
	$\labels(R)$ & Page~\pageref{page:rn suite} &The set of reactions labels\\
	\hline
	$ \Stcoef{r}{x}$ &Page~\pageref{page:rn suite} & The stoichiometry of the species $x$ in the reaction $r$\\
	\hline		
	
	$\ADE{N}$ &Definition~\ref{def:totalode}& The differential equation induced by $N$ \\
	\hline
	$\constr{N}$ &Definition~\ref{def: constraints}& The evolutions constraints induced by $N$ \\
	\hline 
	$\SuccE(N)$ &Definition~\ref{def: euler-succ} &The Euler successor relation with step size $\delta$ for the reaction network $N$ \\
	\hline
	$\SuccS(N) $ &Definition~\ref{def: sign-succ} &The Sign successor relation for the reaction network $N$ \\
	\hline
	$\ATrans{h}{f}$ &Definition~\ref{def:seq} & The set of transitions abstracted from $f$ through $h$ \\
	\hline
	$\solrond(\phi)$ & Page \pageref{page:abstractfunctions}& The set of solution of $\phi$ which store derivatives in new variables \\
	\hline
	$\nextstep(N)$ &Definition~\ref{def: next}& The set of abstract transitions obtained from solutions of $\ADE{N}$ \\
	\hline
	
\end{tabular}

\begin{tabular}{l r l}
	Symbols &Reference & Signification \\	
	\hline	
	$\Absdiff{\homo}{f}(t)$ &Definition~\ref{def: diff}& The differential sign of $f$ at time $t$\\
	
	\hline
	$\slopeset$ & Definition~\ref{def:absimage}& $\{-1,0,1\}\times \{-1,0,1\}$ also denoted $\{-1,0,1\}\times \{\downarrow, \rightarrow, \uparrow\}$ \\
	\hline 
	$\Absimage{f}(t)$ &Definition~\ref{def:absimage} & The pair given by the sign and the differential sign of $f$ at time $t$  \\
	\hline 
	&  & The differential equation such that species in $S_1$ are fixed and\\
	$\fix(N)$ &Definition~\ref{def:fix} & others evolve according the dynamics induced by $N$ \\
	\hline
	$\mathcal{O}$ & Page~\pageref{page:delay}& An arbitrary open set in $\R$ \\
	\hline
	$\cSeg$ &Definition~\ref{def: delays}& The union of delays of $f$\\
	\hline
	$\causalsolution(N)$ &Definition~\ref{def: ccpred}& The set of causal continuous predictions of $N$\\
	\hline
	$\rondcnextstep[\slopeset](N)$ &Definition~\ref{def: cnext}& The set of transitions abstracted from $\causalsolution(N)$ to the domain $\slopeset$\\
	\hline 
	$\slope$ &Definition~\ref{def: slope} & The structure of differential signs\\
	\hline 
	$\downarrow, \rightarrow, \uparrow$ & Page~\pageref{page:diffstructure}&The sign of the differential corresponding to the values $-1, 0$ and $1$\\ 
	\hline
	$(s,\delta)$ & Page~\pageref{page:continuoussucc}& A pair in $\slopeset$\\
	\hline
	$\Continue$ &Definition~\ref{def:continuity}&The set of transitions that may be abstracted from a continuous function \\
	\hline
	$\Continue_{set}$ &Page~\pageref{page:continuoussucc} & The functional representation of the relation $\Continue$\\
	\hline		
	
	$\Succ(N)$&Definition~\ref{def:diff-succ}& The differential sign successor relation for the reaction network $N$\\
	\hline 
\end{tabular}

\section{Existence of solutions}
In this appendix, we establish the existence of solutions for each family of differential equations considered in the paper, and in particular the existence of a continuous causal simulation. A difficulty is that the differential sign abstraction is not always well defined, for instance when a function oscillates infinitely often. To overcome this issue, we restrict our attention to analytic functions which, as we will show later, always admit a differential sign abstraction. Since this notion is not used in the main body of the paper, we briefly recall it here together with some useful properties. For a comprehensive treatment, we refer the reader to \cite{hille2002analytic}.

\begin{definition}[Real-analytic function]
	Let $ \mathcal{O}\subset \mathbb{R}^n$ be an open set. A function
	$f :  \mathcal{O}\to \mathbb{R}$
	is said to be \emph{real-analytic} if for every point $x_0 \in \mathcal{O}$, there exists a neighborhood $V \subset \mathcal{O}$ of $x_0$ and a convergent power series
	$$
	\sum_{\alpha \in \mathbb{N}^n} a_\alpha (x - x_0)^\alpha
	$$
	such that for all $x \in V$,
	$$
	f(x) = \sum_{\alpha \in \mathbb{N}^n} a_\alpha (x - x_0)^\alpha.
	$$
	Here $\alpha = (\alpha_1,\dots,\alpha_n)$ is a multi-index and
	$$
	(x - x_0)^\alpha = \prod_{i=1}^n (x_i - x_{0,i})^{\alpha_i}.
	$$
\end{definition}		

The existence of analytic solutions is established by the Cauchy--Kovalevskaya theorem.

\begin{theorem}[Cauchy--Kovalevskaya, ODE case]\label{th: cauchy-K}
	Let $ \mathcal{O}\subset \mathbb{R} \times \mathbb{R}^n$ be an open set, 
	$
	F :  \mathcal{O}\to \mathbb{R}^n
	$
	be an analytic function,	
	and $(t_0, x_0) \in \mathcal{O}$ be a time point and the state of the system . Then there exists a neighbourhood $I$ of $t_0$
	and a unique function
	$
	f : I \to \mathbb{R}^n
	$
	which is real-analytic and satisfies
	$$
	\begin{cases}
		\dot{f}(t) = F(t, f(t)), \\
		f(t_0) = x_0.
	\end{cases}
	$$	
	Moreover, the solution $f(t)$ is analytic on $I$.
\end{theorem}

\begin{property}[Linear combination]\label{prop: linear}
	Linear combinations of analytic functions are still analytic functions. 
\end{property}

\begin{property}[Existence of differential abstraction]\label{prop:existIsd}
	Let  $\mathcal{O}\subseteq \R$ be an open set and $f:\mathcal{O} \to \R$ be an analytic function.
	For all $t \in \mathcal{O}$, $\Absdiff{\homo}{f} (t)$ exists.
\end{property}
\begin{proof}
	
	By definition of analytic functions, it exists a sequence $(\alpha_i)_{i\in \Nat} \in \R^\Nat$
	such that for all $x\in \mathcal{O}$, 
	$$	f(x) = \sum_{i \in \mathbb{N}} \alpha_i (x - t)^i.$$
	
	Let us recall that 
	$$\Absdiff{\homo}{f}(t)= \lim_{x\to t^+} \homo(f(x)-f(t))$$
	and 
	\begin{align}\label{appendice-1}
		\lim_{x\to t^+} \left( \sum_{i > 0} \alpha_{i}(x - t)^i \right) = 0.
	\end{align}
	We can deduce that $f(t) = \alpha_0$, and
	$$\Absdiff{\homo}{f}(t)= \lim_{x\to t^+} \homo \left(\sum_{i >0} \alpha_i (x - t)^i\right).$$
	If for all $i\in \Nat\setminus\{0\}$, $\alpha_i=0$, then $f$ is constant. Therefore, we have $\Absdiff{\homo}{f} (t) = 0$. 
	Otherwise, let $k\in \Nat\setminus\{0\}$ be the smallest integer non-zero such that $\alpha_k \neq 0$. According to (\ref{appendice-1}), it exists $ \epsilon >0$ such that $ \forall x \in ]t,t+\epsilon[$, we have:
	\begin{align*}
		\left| \sum_{i >0} \alpha_{i+k} (x - t)^i \right| &< \left| \alpha_k \right|\\
		\Rightarrow~~ \left| \sum_{i >0} \alpha_{i+k} (x - t)^ {i+k}  \right| &< \left| \alpha_k (x - t)^ {k}\right|\\
		\Rightarrow	~~  h_{\s} \left(\alpha_k (x - t)^ {k} + \sum_{i> 0 } \alpha_{i+k} (x - t)^ {i+k}\right) &= h_{\s}\left(\alpha_k(x - t)^ {k}\right)\\
		\Rightarrow ~~  h_{\s} (f(x)-f(t)) &= h_{\s}\left(\alpha_k(x - t)^ {k}\right)\\
		\Rightarrow ~~  h_{\s} (f(x)-f(t)) &= h_{\s}(\alpha_k)
	\end{align*}
	To conclude, in the case of analytic functions, $\Absdiff{\homo}{f}$ is defined on the domain of definition. 
\end{proof}
In order to simplify the proof of existence of solution of $\ADE{\cdot}$, in the following
we will consider the differential equation where expressions appear in the place of reaction labels:
$$ 
\mathsf{ODE}(N):=  \bigwedge_{x \in S}  \left(\dot{x} \EQUAL \sum_{r\in \labels (R)} \Stcoef{r}{x}\mult e_r  \right).
$$ To exploit $\mathsf{ODE}(N)$, we need the following property of substitution.  
\begin{property}[Substitution]\label{prop: subs}		
	
	Let $N$ be a reaction network, and  $f\in \sol(\mathsf{ODE}(N))$ a solution of its ordinary differential equation.
	The function $f': (S\cup\labels(R))\to (\R\to\R)$ such that $\restr{f'}{S}= f$ and 
	for all $r\in \labels(R)$, $f'(r) = \fexpr[e_r] \circ f$, is a solution of $\ADE{N}$,i.e. $ f'\in \sol(\ADE{N})$.	
\end{property}
\begin{proof}[SKETCH]
	The interpretation of an expression in the structure of real functions give always at most a singleton. Therefore, the interpretation of $\EQUAL$ 
	($\crochet{e_1\EQUAL e_2}{\struct,\alpha}$ is true only if $\crochet{e_1}{\struct,\alpha} \cap \crochet{e_2}{\struct,\alpha}\neq \emptyset$ ) is equivalent to the standard equality, and we can substitute.
\end{proof}
\begin{property}[Existence of analytic solutions of $\ADE{N}$]\label{prop: exists ADE}
	Let  $ N = (S,R) $ be a reaction network such that for all $r\in \labels(R) $, $\fexpr[e_r]$ is analytic. Let $t_0\in \R_+$ be a time point and $\eta \in \R_+^S$ be a state of the system.
	There exists a neighbourhood $I$ of $t_0$ and a unique function $f: I \to \R^S$ that satisfies:
	$$f \in \sol(\ADE{N}) ~ \land  ~ \restr{f(t_0)}{S} = \eta.$$
	Moreover, $f$ is analytic. 
\end{property}
\begin{proof}
	We will work with the differential equation $\mathsf{ODE}(N)$, as from its solutions we can easily construct the solution of $\ADE{N}$ (Proposition~\ref{prop: subs}).
	If each expression in the reaction network is analytic then for all $x\in S$, $\fexpr[\sum_{r\in \labels (R)} \Stcoef{r}{x}\mult e_r]$ is analytic too by Proposition~\ref{prop: linear}. Therefore, by Theorem~\ref{th: cauchy-K}, there exist a neighbourhood $I$ of $t_0$ and a unique function $f$ such that $f \in \sol(\ADE{N})$ and $\restr{f(t_0)}{S} = \eta.$
\end{proof}
\begin{property}[Existence solutions of $\fix(N)$]\label{prop: exists FIX}		
	
	Let  $ N = (S,R) $ be a reaction network such that for all $r\in \labels(R) $, $\fexpr[e_r]$ is analytic.
	Let $t_0 \in \R_+$ be a time point, $\eta \in \R_+^S$ be a state of the system and $S_1\subseteq S$ be a set of fixed species.
	There exists a neighbourhood $I$ of $t_0$ and a unique function $f: I \to \R^S$ that satisfies:
	$$f \in \FIX{S_1}{N} ~\land ~ \restr{f(t_0)}{S} = \eta.$$
	Moreover, $f$ is analytic. 
\end{property}
\begin{proof}[SKETCH]
	Given the network $N$, we can build a network $N'$ such that $\fix(N) = \ADE{N'}$. 
	$N'$ is obtained  by removing species of $S_1$ from multisets of reactants and products of each reaction, without changing the kinetic expressions.  	
	Therefore, according to Property~\ref{prop: exists ADE}, such a solution exists. 
	
\end{proof}

\begin{property}[Existence solutions in $\causalsolution(N)$]\label{prop: exists CSOL}	
	Let  $ N = (S,R) $ be a reaction network such that for all $r\in \labels(R) $, $\fexpr[e_r]$ is analytic.
	Let $t_0 \in \R_+$ be a time point, and $\eta \in \R_+^S$ be a state of the system. There exists a neighbourhood $I$ of $t_0$ and a unique function $f: I \to \R^S$ that satisfies:
	$$f\in \causalsolution(N) ~\land~ \restr{f(t_0)}{S} = \eta.$$ 
	Moreover, $f$ is analytic. 
\end{property}
\begin{proof}
	Let us consider $f'\in \rond{sol}(\ADE{N})$, such that
	$$
	\restr{f'(t_0)}{S} = \eta
	$$
	and assume that $f'$ is analytic (whose existence follows from Property~\ref{prop: exists ADE}).
	
	By Property~\ref{prop:existIsd}, $\Absimage{f'}$ is defined on the same domain as $f'$.
	
	If $f'$ is a causal match on an open neighborhood of $t_0$, then it is continuous and causal, trivially consistent, and has no delay on this neighborhood. Therefore, it is a causal continuous prediction on this set.
	
	Otherwise, $f'$ is not a causal match at $t_0$. Let $S_1$ be the set of species whose derivative is zero at $t_0$. Let us consider
	$$
	f'' \in \FIX{S_1}{N}
	$$
	such that
	$$
	\restr{f''(t_0)}{S} = \eta
	$$
	and assume that $f''$ is analytic (whose existence follows from Property~\ref{prop: exists FIX}).
	By the analytic characteristic of $f'$, it exists $\epsilon > 0$ such that there is no causal mismatch for $f'$ on $]t_0- \epsilon, t_0[$. Therefore we consider the function $f$ such that $\restr{f}{]t_0- \epsilon, t_0[} = \restr{f'}{]t_0- \epsilon, t_0[} $ and for all $t\in \DomF{f''}\cap [t_0,+\infty[$, $f(t)= f''(t)$.
	We now show that there exists $t_1 > t_0$ such that
	$$
	\restr{f}{]t_0-\epsilon,t_1[} \in \causalsolution(N).
	$$
	
	\begin{itemize}
		\item\emph{[continuity]} Since $f'$  and $f''$ are both analytic and $f'(t_0) = f''(t_0)$,  $f$ is continuous.
		
		\item  \emph{[admissible delays]} Let $x\in S\cup \rond{S}\cup \labels(R)$ be a variable of the network. Since $f''_x$ is analytic, $\Absdiff{\homo}{f''_x}$ is defined on $\DomF{f''}$. By definition of the abstraction of a differential, if $\Absdiff{\homo}{f''_x}(t_0)$ exists, then there exists $\epsilon > 0$ such that for all $t \in [t_0,t_0+\epsilon[$,
		$$
		\Absdiff{\homo}{f''_x}(t) = \Absdiff{\homo}{f''_x}(t_0).
		$$
		
		Moreover, if $\homo \circ f''_x(t_0) \in \{-1,1\}$, then the sign remains constant on some neighborhood of $t_0$. 
		
		Otherwise $\homo \circ f''_x(t_0)=0,$ and
		the sign may change arbitrarily close to $t_0$.
		Therefore, for each variable $x$, one of the following holds:
		\begin{enumerate}
			\item there exists $t_{1,x} > t_0$ such that $\Absimage{f''_x}$ is constant on $[t_0,t_{1,x}[$;
			\item $f''_x(t_0)=0$ and $f''_x$ is not constant. Since $\Absdiff{\homo}{f''_x}(t_0)$ exists, it exists $t_{1,x}>t_0$, such that $f''_x$ is either strictly positive or strictly negative on $]t_0,t_{1,x}]$.
		\end{enumerate}
		
		In both cases, there exists $t_{1,x}> t_0$ such that $\Absimage{f''_x}$ is constant on $]t_0,t_{1,x}[$. Let $t_1$ be the minimum of all $t_{1,x}$. On the interval $]t_0,t_1]$,  $\Absimage{f_x}$ for each species $x$ is constant.
		Moreover, we chose $S_1$ as the set of species with a null derivative at $t_0$.

		\item \emph{[consistency]} By construction $f$ is consistent. 
		
		\item \emph{[causality]} We choose $t_0-\epsilon$ in a way that there is no mismatch on 
		$]t_0-\epsilon,t_0[$. Since we chose $S_1$ as the set of species with null derivative, 
		$f$  is a causal match at $t_0$.  Finally on $]t_0,t_1]$, every species has a constant differential abstraction therefore there is no mismatch on $]t_0,t_1[$.

	\end{itemize}
	To conclude $
	\restr{f}{]t_0-\epsilon,t_1[} \in \causalsolution(N).
	$
\end{proof}
	
\end{document}